\documentclass[a4paper]{article}
\usepackage[utf8]{inputenc}
\usepackage{amsmath}
\usepackage{amssymb}
\usepackage{amsfonts}
\usepackage{amsthm}
\usepackage{graphicx}
\usepackage{color}
\usepackage{cite}
\usepackage{tikz}
\usepackage{commands-tam}
\usepackage{subcaption}
\usepackage{xspace}
\newcommand\ignore[1]{}
\definecolor{cola}{RGB}{255,128,9}\definecolor{colb}{RGB}{0,136,218}\definecolor{colc}{RGB}{0,205,96}\definecolor{cold}{RGB}{255,192,22}\definecolor{cole}{RGB}{113,1,184}
\newtheorem{theorem}{Theorem}[section]

\newtheorem{lemma}[theorem]{Lemma}

\newtheorem{definition}[theorem]{Definition}

\newcommand\resp{respectively\xspace}

\newcommand\vect[1]{\vec{#1}}
\newcommand{\m}[2]{\GenericWarning{}{LaTeX Warning: #1: #2}\marginpar{\tiny \sloppy #1: #2}}

\newenvironment{reptheorem}[1]{
  \hfill

  \noindent{\bf Theorem~\ref{#1}.\hspace*{0.5em}}\it
}{\hfill}

\author{
  Florent Becker%
     \thanks{
       Univ. d'Orléans, INSA Centre Val de Loire, LIFO EA 4022, FR-45067 Orléans, France,
       \href{mailto:florent.becker@univ-orleans.fr}{florent.becker@univ-orleans.fr}.
       Supported in part by ANR Grant / Avec le soutien du programme ANR ``QuasiCool''.
     }
  \and
  Pierre-\'Etienne Meunier\thanks{Department of Computer Science, Aalto University, Helsinki, Finland, \href{mailto:pierre-etienne.meunier@aalto.fi}{pierre-etienne.meunier@aalto.fi}. This research was supported in part by National Science Foundation Grant CCF-1219274.}
  }

\title{It's a Tough Nanoworld: in Tile Assembly, Cooperation is not (strictly) more Powerful than Competition}
\date{}
\begin{document}
\maketitle

\begin{abstract}

We present a strict separation between the class of ``mismatch free'' self-assembly systems and general aTAM systems. Mismatch free systems are those systems in which concurrently grown parts must always agree with each other.

Tile self-assembly is a model of the formation of crystal growth, in which a large number of particles concurrently and selectively stick to each other, forming complex shapes and structures. It is useful in nanotechnologies, and more generally in the understanding of these processes, ubiquitous in natural systems.

The other property of the local assembly process known to change the power of the model is cooperation between two tiles to attach another. We show that disagreement (\emph{mismatches}) and cooperation are incomparable: neither can be used to simulate the other one.

The fact that mismatches are a hard property is especially surprising, since no known, explicit construction of a computational device in tile assembly uses mismatches, except for the very recent construction (FOCS 2012) of an intrinsically universal tileset, i.e. a tileset capable of simulating any other tileset up to rescaling. This work shows how to use intrinsic universality in a systematic way to highlight the essence of different \emph{features} of a model.

Moreover, even the most recent experimental realizations do not use competition, which, in view of our results, suggests that a large part of the natural phenomena behind DNA self-assembly remains to be understood experimentally.

\end{abstract}

\section{Introduction}

The formidable biological diversity of living organisms is due for a large part to the amazing richness of the computational processes at the core of natural processes, at the molecular level. A largely studied molecular interaction paradigm is \emph{cooperation} between molecules, in fields such as proteomics~\cite{pmid10189717}. However, in tile self-assembly -- a computing model inspired by such processes -- \emph{competition} seems to be at least as important as cooperation. In this paper, we investigate the relations between these two interaction paradigms, and conclude that both are necessary for self-assembly to retain its full computational power.

Tile self-assembly is a model of molecular growth introduced by Winfree~\cite{Winf98} to study the molecular engineering possibilities offered by the components designed by Seeman~\cite{Seem82} using DNA. This model studies the autonomous assembly of independent atomic components, in an unsupervised way.  Using these ideas and their later developments, researchers have been able to build a number of structures, ranging from regular arrays~\cite{WinLiuWenSee98} to fractal structures~\cite{RoPaWi04,FujHarParWinMur07}, smiling faces~\cite{RothOrigami,wei2012complex}, DNA tweezers~\cite{yurke2000dna}, logic circuits~\cite{seelig2006enzyme,qian2011scaling}, neural networks~\cite{qian2011neural}, or molecular robots\cite{DNARobotNature2010}. These examples not only demonstrate our ability to craft nano-things in bottom-up processes, as opposed to traditional top-down crafting processes; they also provide a strong link between the theoretical study of tile self-assembly, and actual processes occurring in nature, as observed in wet-lab experiments.

In tile self-assembly, we consider the assembly of non-flippable, non-rotatable square tiles, with a \emph{glue color} and an integer \emph{glue strength} on each side, to an existing assembly. We consider a finite amount of \emph{tile types}, with an infinite supply of each type. The dynamics starts from a \emph{seed} assembly and proceeds asynchronously and non-deterministically, one tile at a time. At each step, a tile can stick to an existing assembly if the sum of glue strengths with neighboring tiles matching its glue colors is at least a parameter of the model called the \emph{temperature}. This means that their may be mismatches between adjacent tiles; the only requirement is that there be enough glue strength for each tile to attach to the assembly.

This model is similar to Wang tilings~\cite{Wang61,Wang63,Berger65,Robinson71}, essentially augmented with a mechanism for sequential, asynchronous growth, and in which mismatches are allowed. Despite its simplicity, it is very expressive, as it can simulate the behavior of any Turing machine~\cite{RotWin00}, and produce arbitrary connected shapes with a number of tile types within a logarithmic factor of their size~\cite{SolWin05,Winfree98simulationsof}.

\paragraph{Intrinsic simulations} Our purpose in this work is to compare the dynamics of different models of assembly. A first objection to this, is that adding new tile types can change dramatically the computational power of our model: therefore, our comparisons would be not be very useful in the qualitative study of molecular processes if they consisted only of statements of the form \emph{``model $X$ can do something with one billion tile types, that model $Y$ cannot do with ten tile types''}.

Getting round this objection is the purpose of \emph{intrinsic simulations}, a notion of dynamic simulation \emph{up to rescaling}. This idea originated in cellular automata, and has given rise to a large literature~\cite{bulkingI,bulkingII,arrighi2012intrinsic,Ollinger08,goles-communicationcomplexity}, and has also been studied in Wang tiling~\cite{LafitteW07,LafitteW08,LafitteW09}. More recently, it has been adapted to tile assembly~\cite{USA,IUSA,Meunier-2014,woods2013intrinsic,2HAMIU}, the main difficulties of this adaptation being that tile assembly is \emph{asynchronous}, and that the geometry of assemblies plays an important role, since in most models, a tile that has grown cannot detach from the assembly nor change its type, contrarily to the case of cellular automata.

Intrinsic simulations shall not be seen as a way to bypass the difficulty of proving negative results about tile self-assembly: indeed, a major achievement of this line of research is the existence of an intrinsically universal tileset~\cite{IUSA}, i.e. a tileset that can simulate the behavior of any other modulo rescaling. This notion can therefore be used to separate different models or classes, by proving -- as done in \cite{Meunier-2014}, and as we do in the present work -- that a model cannot simulate arbitrary tile assembly systems, whereas arbitrary tile assembly systems can.
Intrinsic universality has also allowed to identify important challenges such as the existence an intrinsically universal tileset with a single (rotatable) polygonal tile~\cite{Demaine-2014}, giving new insights on the long-standing open problem of the existence of an aperiodic tiling of the plane, with a single (rotatable) tile.

Moreover, it has been shown that intrinsic universality is distinct from Turing universality: for instance, the three-dimensional variant of non-cooperative tile assembly can simulate Turing machines~\cite{Cook-2011}, but cannot intrinsically simulate arbitrary tile assembly systems~\cite{Meunier-2014}.

\paragraph{Mismatch-free and locally consistent systems}
An important class of tile assembly system, that has been extensively studied in the context of \emph{non-cooperative} (i.e. temperature 1) tile assembly, is the class of systems whose productions never have mismatches between adjacent tiles. It has been shown~\cite{Manuch-2010,Doty-2011} that these systems are only capable of assembling semi-periodic shapes. However, a long-standing open problem of tile assembly is the decidability of this condition for non-cooperative systems. Moreover, in the three-dimensional generalization of tile assembly, mismatches are an essential aspect of the computational capabilities of the non-cooperative model~\cite{Cook-2011}.

At temperature at least 2, the situation is quite different, since almost all known constructions, with the exception of the intrinsically universal tilesets that have been built~\cite{IUSA,Demaine-2014}, never make mismatches between adjacent tiles. Moreover, an intermediate result on the road to intrinsic universality was proven for the restricted class of \emph{locally consistent} tile assembly systems~\cite{USA}, that are mismatch-free systems in which it is additionally required that all tiles attach with the sum of glue strengths \emph{exactly} equal to the temperature, whereas in the regular model, tiles can also attach with strengths summing up to more than the temperature.

Moreover, as stated earlier in the introduction, most of the behaviors observed in natural systems are studied with \emph{cooperation} in mind; the question is therefore asked, of whether these assumptions (the mismatch-free property, and local consistency) really weaken the model. More precisely, the construction in~\cite{IUSA} of an intrinsically universal tileset relies crucially on mismatches: Sections 5.2.1, 5.3.1, 5.4.1 are typical examples where competition seems unavoidable.

\subsection{Main results}

Our results show the ``feature-optimality'' of these constructions, by proving that cooperation and redundancy \emph{are actually} necessary to retain the full power of the model. They also show the relevance of intrinsic universality, as opposed to only computational universality, to study natural systems: indeed, in tile assembly, computational universality can be achieved using a very weak subset of the model's features, and therefore fails to capture key elements of processes occurring in nature.

Finally, our results call for a deeper discussion with experimentalists: indeed, the most recent experimental implementations of computational processes using tiles~\cite{evans2014crystals} do not use competition nor glue redundancy. If these are really happening in nature, our work gives examples of phenomena allowed by tile assembly, that ought to be demonstrated experimentally to achieve a full understanding of crystal growth.

More precisely, the features we study are:
\begin{itemize}
\item cooperation, through increases of the \emph{temperature};
\item competition, through the presence of mismatches;
\item redundancy, through attachment strength exceeding the temperature;
\item and geometry, through the dimension of the space the assembly takes place in.
\end{itemize}

The two questions we aim to answer are, for each class, does it have a \emph{complete} system, i.e. a system that is intrinsically universal for systems that are in the class itself.

The following is known:
\begin{itemize}
\item non-planarity can sometimes replace cooperation: temperature $1$ systems in three dimensions can simulate planar temperature $2$ zig-zag systems, which can simulate arbitrary Turing machines~\cite{Cook-2011};
\item cooperation is a matter of getting the temperature to $2$, more is superfluous: there is an intrinsically universal tileset at temperature $2$, for any dimension, when mismatches and over-attachments are allowed \cite{IUSA};
\item cooperation cannot be simulated by temperature $1$ systems \cite{Meunier-2014};
\item \emph{locally consistent systems}, the ones that exhibit neither competition nor redundancy have a complete tileset, which works at temperature $2$.
\end{itemize}

We complete this study be showing the following:
\begin{theorem}
\label{thm:nomismatches}
Competition cannot be simulated by cooperation: there is a temperature $1$ system with mismatches which cannot be simulated by any system without mismatches. This remains true whether or not redundancy is allowed.
\end{theorem}

\begin{theorem}
\label{thm:local-vs-mismatch-free}
Redundancy cannot be simulated either: there is a temperature $1$ system with attachments exceeding the temperature which cannot be simulated at any temperature by \emph{locally consistent systems}, i.e. without redundancy or competition.
\end{theorem}

\begin{theorem}
\label{thm:full-hotel}
There is an intrinsically universal tileset for the class of systems without mismatches in $\Z^3$ (the question remains open in $\mathbb{Z}^2$).
\end{theorem}

The complexity landscape of self-assembly we get is described on Figure \ref{fig:landscape}.

\begin{figure}
  \centering
  \scalebox{0.6}{\begin{tikzpicture}[>=latex,line join=bevel,]
  \pgfsetlinewidth{1bp}
\pgfsetcolor{black}
  \draw [->] (223.91bp,94.975bp) .. controls (207.83bp,80.332bp) and (183.88bp,58.51bp)  .. (157.89bp,34.842bp);
  \definecolor{strokecol}{rgb}{0.0,0.0,0.0};
  \pgfsetstrokecolor{strokecol}
  \draw (241.5bp,65bp) node {simulates \cite{SolWin05,Cook-2011}};
  \draw [->] (142.32bp,281.83bp) .. controls (141.97bp,249.45bp) and (141.21bp,179.09bp)  .. (140.68bp,130.12bp);
  \draw (159.5bp,206bp) node {{\color{colc}{$\supset$}} \cite{Meunier-2014}};
  \draw [->] (285.65bp,120.61bp) .. controls (299.03bp,120.4bp) and (309.5bp,117.53bp)  .. (309.5bp,112bp) .. controls (309.5bp,107.99bp) and (303.98bp,105.37bp)  .. (285.65bp,103.39bp);
  \draw (321bp,112bp) node {\cite{USA}};
  \draw [->] (157.66bp,284.62bp) .. controls (173.14bp,269.94bp) and (197.4bp,246.93bp)  .. (222.83bp,222.82bp);
  \draw (251bp,253bp) node {$\supseteq$, \color{cole}{$\supset$ (Theorem~\ref{thm:nomismatches})}};
  \draw [->] (124.36bp,286.5bp) .. controls (115.22bp,280.03bp) and (104.07bp,271.85bp)  .. (94.5bp,264bp) .. controls (81.087bp,253bp) and (66.792bp,239.96bp)  .. (47.86bp,222.03bp);
  \draw (111.5bp,253bp) node {{\color{colc}{$\supset$}} \cite{Meunier-2014}};
  \draw [->] (61.307bp,214.56bp) .. controls (73.075bp,215.1bp) and (83bp,212.24bp)  .. (83bp,206bp) .. controls (83bp,201.71bp) and (78.309bp,199.02bp)  .. (61.307bp,197.44bp);
  \draw (94.5bp,206bp) node {\cite{Meunier-2014}};
  \draw [->] (49.852bp,190.22bp) .. controls (67.67bp,175.04bp) and (95.568bp,151.28bp)  .. (123.55bp,127.44bp);
  \draw (109bp,159bp) node {$\supseteq$};
  \draw [->] (239.88bp,187.7bp) .. controls (240.16bp,174.46bp) and (240.57bp,155.95bp)  .. (241.13bp,130.23bp);
  \draw (289bp,159bp) node {$\supseteq$, \color{cole}{${\supset}$ (Theorem~\ref{thm:local-vs-mismatch-free})}};
  \draw [->,dashed] (140.5bp,93.696bp) .. controls (140.5bp,80.46bp) and (140.5bp,61.947bp)  .. (140.5bp,36.227bp);
  \draw (148.5bp,65bp) node {?};
  \draw [->] (166.57bp,308.45bp) .. controls (177.74bp,309.38bp) and (187.5bp,306.56bp)  .. (187.5bp,300bp) .. controls (187.5bp,295.59bp) and (183.1bp,292.88bp)  .. (166.57bp,291.55bp);
  \draw (199bp,300bp) node {\cite{IUSA}};
  \draw [->] (27.895bp,188.06bp) .. controls (22.675bp,165.25bp) and (16.781bp,124.05bp)  .. (32.5bp,94bp) .. controls (46.572bp,67.095bp) and (75.176bp,47.902bp)  .. (108.45bp,31.321bp);
  \draw (65.5bp,112bp) node {simulates \cite{Cook-2011}};
  \draw [->] (278.25bp,214.62bp) .. controls (291.17bp,214.62bp) and (301.5bp,211.75bp)  .. (301.5bp,206bp) .. controls (301.5bp,201.87bp) and (296.16bp,199.22bp)  .. (278.25bp,197.38bp);
  \draw (352bp,206bp) node {\color{cole}{In 3D (Theorem~\ref{thm:full-hotel})}};
\begin{scope}
  \definecolor{strokecol}{rgb}{0.0,0.0,0.0};
  \pgfsetstrokecolor{strokecol}
  \draw (239.5bp,206bp) node {No mismatches};
\end{scope}
\begin{scope}
  \definecolor{strokecol}{rgb}{0.0,0.0,0.0};
  \pgfsetstrokecolor{strokecol}
  \draw (140.5bp,18bp) node {Turing machines};
\end{scope}
\begin{scope}
  \definecolor{strokecol}{rgb}{0.0,0.0,0.0};
  \pgfsetstrokecolor{strokecol}
  \draw (241.5bp,112bp) node {Locally consistent};
\end{scope}
\begin{scope}
  \definecolor{strokecol}{rgb}{0.0,0.0,0.0};
  \pgfsetstrokecolor{strokecol}
  \draw (32.5bp,206bp) node {3D, $\tau=1$};
\end{scope}
\begin{scope}
  \definecolor{strokecol}{rgb}{0.0,0.0,0.0};
  \pgfsetstrokecolor{strokecol}
  \draw (142.5bp,300bp) node {$\tau=2$};
\end{scope}
\begin{scope}
  \definecolor{strokecol}{rgb}{0.0,0.0,0.0};
  \pgfsetstrokecolor{strokecol}
  \draw (140.5bp,112bp) node {2D, $\tau=1$};
\end{scope}
\end{tikzpicture}}
  \caption{The landscape of behavior complexity in self-assembly. The comparisons in purple are the main results of this paper. Models with a loop have a complete tileset, i.e. one that is intrinsically universal for the class itself. The dashed arrow represents a conjectured comparison. Very few separation results between classes were previously known; on this diagram, they are written in green.}
  \label{fig:landscape}
\end{figure}
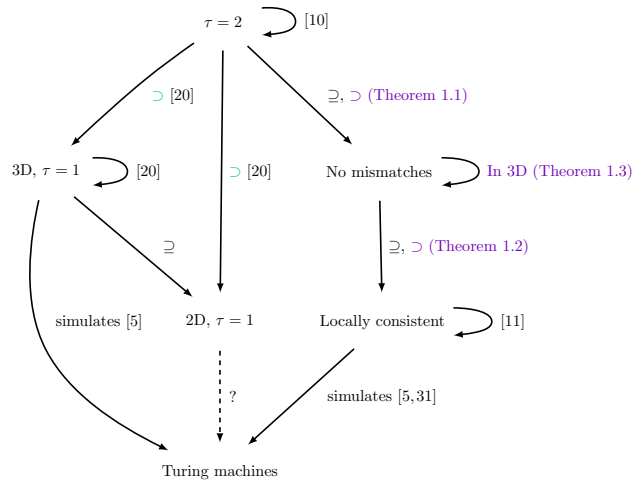

Known constructions such as found in \cite{IUSA} use competition a lot to ensure that choices are unique. Theorem \ref{thm:nomismatches} shows that this is in fact necessary, at least if the simulated system has mismatches. In other words, cooperation and competition are actually different phenomena which cannot simulate each other.

Theorem \ref{thm:full-hotel} shows that this tension between competition and cooperation is weaker in three-dimensional systems, as cooperation can be simulated without competition. This parallels the result of \cite{Patitz-2011} showing that in three dimensions, some aspects of cooperation such as Turing-completeness can be simulated using competition.

\subsection{Key ideas}

The main technical contributions of the paper are Lemma \ref{lem:bisimilarity} (called the ``\emph{bisimilarity lemma}''), and the cutspace idea found in theorems \ref{thm:nomismatches} and \ref{thm:local-vs-mismatch-free}. These are two powerful new tools for establishing negative results in self-assembly. There is a wealth of positive results (constructions) in self-assembly literature, but negative results are rare in comparison. One of the reasons for this is the difficulty of analyzing the dynamics of self-assembly systems, mainly because the intuitions from usual one-dimensional computing models (such as automata and Turing machines) are often false when more complex geometries are involved (even $\mathbb{Z}^2$).

\paragraph{The bisimilarity lemma} (Lemma \ref{lem:bisimilarity}) is a very powerful tool for exploring the limits of tile assembly dynamics. Fundamentally, it is a ``pumping lemma'' for tile assembly systems. When manipulating such systems, it comes quite naturally to see that long and narrow part of the assembly act as wires with limited capacity. A few proofs have relied on that characteristic to prove that some part of an assembly can not possibly get enough information to behave as wanted. Yet, these proofs were tied to the geometry of such assemblies and what information they bore. The bisimilarity lemma gives a more comprehensive account of these limitations: on such a wire, communication might be bidirectionnal and tile additions may bring information in either direction. On a ``wire'' of finite width, only a small number of ``communications'' can be performed between distinct parts of the assembly.

 The heavy lifting for this lemma lies not in its proof, but in the distinction between the \emph{diplomatic set} representing what can be seen of the assembly through some cut-set of the grid and the \emph{policy set} of what actually happens there. A simple analogy for this lemma is: in order to know what happens in the world outside Las Vegas, you only need to know when and in what state people going into Clark County will come back out. What actually happens in there is of no consequence to the outside world\footnote{More to the point, Vegas does not need an actual ``real world'' around it, it just needs cable TV in order to know how the world reacts to what it sees of Vegas.}.

This result is a generalization of a previous result called the \emph{window movie lemma}~\cite{Meunier-2014}: indeed, the window movie lemma allows to ``pause'' the dynamics of self-assembly at some step, ``move'' parts of the assembly around, and then resume the dynamics. The \emph{bisimilarity lemma} allows to repeat this operation: we can pause the dynamics, move parts by some vector $\vect{v}$, resume the dynamics for some time, and then move the parts back in place, with the current production. This operation can be repeated many times \emph{with the same translation vector $\vect{v}$}.

\paragraph{Mismatchless systems and the cutspace.} Once we have the right tool to bound the amount of information that is exchanged between different parts of the assembly, we also need a qualitative limit on what can be done without mismatches. Experimenting with such systems yields a first insight: without mismatches, concurrent assemblies ``pass through each other'', and each of them will act is if it were the first one to come at a given spot. This is most acute whenever we have the archetypical race condition: two lines growing concurrently towards some point, the first come blocking the second from passing through.

 To turn this intuition into a formal tool, we introduce the idea of \emph{cutspaces}. The idea is to consider a tile assembly system, and to change the graph on which the assembly takes place. By doing so, we can observe what happens when all concurrent sub-assemblies are ``\emph{simultaneously} the first'' to come at each race position. The key idea is to design a graph where each part of the assembly ``sees'' a different copy of the race position. Thus, they really all come first, in their own ``parallel universe'' where they have won the race. We then need to design a strategy to collapse these universes back into the same real spot and observe that without mismatches and without backdoor communication (ruled out through the bisimilarity lemma), they must all behave as they did on their own. Finally, all of them proceed at once as if they had won, contradicting the hypothesis that they are simulating a mismatch.

\paragraph{Three dimensional constructions.} Theorem \ref{thm:full-hotel}, an intrinsically universal tileset for mismatchless systems, relies on the fact that in non-planar grids, it is possible to wire up simulations without much limitation. In planar simulations of self assembly, it is crucial that either the set of inputs can be defined unambiguously, which rules out competition and redundancy in the simulated system, or that a race mechanism allows one of the possible sets of inputs to be selected as the winner, which requires competition in the simulator. In Theorem \ref{thm:full-hotel}, we get rid of this constraint by simulating (up to rescaling) all possible sets of input sides for a tile simultaneously.
Since we are simulating systems without mismatches, the representations of the possible outputs will always agree with neighboring tiles, even if these come afterwards.

\section{Definitions and preliminaries}

\label{definitions}
\subsection{Tile self-assembly}

We begin by defining the two-dimensional (\resp three-dimensional) abstract tile assembly model. A \emph{tile type} is a unit square or hexagon (\resp cube) with four or six (\resp six) sides, each consisting of a glue \emph{label} and a nonnegative integer \emph{strength}.
In the square grid of the plane, we call a tile's sides north, east, south, and west, respectively, according to the following picture:

We assume a finite set $T$ of tile types, but an infinite supply of copies of each type. An \emph{assembly} is a positioning of the tiles on a regular graph $G=(V,E)$ of degree four %
called the \emph{space graph}, that is, a partial function $\alpha:V\dashrightarrow T$.

In the case of the square %
grid in two dimensions, and of the cubic grid in the plane, $G$ is the Cayley graph of a commutative group, with generators $(1,0)$ and $(0,1)$ in the square grid of $\mathbb{Z}^2$, %
and $(1,0,0)$, $(0,1,0)$ and $(0,0,1)$ in the cubic grid of $\mathbb{Z}^3$. Therefore, in the rest of this paper, when there is no ambiguity, we will use the elements of the groups both as vectors and as points of the space. We will however use other kinds of graphs to formalize the idea of \emph{cutspaces}.

We say that two tiles in an assembly \emph{interact}, or are \emph{stably attached}, if the glue labels on their abutting side are equal, and have positive strength. An assembly $\alpha$ induces a weighted \emph{binding graph} $B_\alpha=(V_\alpha,E_\alpha)$, which is a subgraph of $G$, where $V_\alpha=\dom{\alpha}$, and there is an edge $(a,b)\in E_\alpha$ if and only if $a$ and $b$ interact, and this edge is weighted by the glue strength of that interaction.  The assembly is said to be $\tau$-stable if any cut of $G$ has weight at least $\tau$.

A \emph{tile assembly system} is a triple $\mathcal{T}=(T,\sigma,\tau)$, where $T$ is a finite tile set, $\sigma$ is called the \emph{seed}, and $\tau$ is the \emph{temperature}.

Given two $\tau$-stable assemblies $\alpha$ and $\beta$, we say that $\alpha$ is a \emph{subassembly} of $\beta$, and write $\alpha\sqsubseteq\beta$, if $\dom{\alpha}\subseteq \dom{\beta}$ and for all $p\in \dom{\alpha}$, $\alpha(p)=\beta(p)$.  We also write $\alpha\rightarrow_1^{\mathcal{T}}\beta$ if we can get $\beta$ from $\alpha$ by the binding of a single tile, that is, if $\alpha\sqsubseteq \beta$ and $|\dom{\beta}\setminus\dom{\alpha}|=1$. We note such a binding $t@z$ (read ``tile $t$ at position $z$''), where $\{z\} = \dom{\beta}\setminus\dom{\alpha}$, and $t=\beta(z)$.
 We say that $\gamma$ is \emph{producible} from $\alpha$, and write $\alpha\rightarrow^{\mathcal{T}}\gamma$ if there is a (possibly empty) sequence $\alpha=\alpha_1,\ldots,\alpha_n=\gamma$ such that $\alpha_1\rightarrow_1^{\mathcal{T}}\ldots\rightarrow_1^{\mathcal{T}}\alpha_n$.

For any $X\subseteq\mathbb{N}$, we write $(U_n)_{n\in X}$ for ``the sequence defined for all $i\in X$ by $U_i$'', and if $U$ is a finite sequence and $V$ is an arbitrary sequence, we alse write $U.V$ for the concatenation of the two sequences, i.e. the sequence $(W_n)_{n\geq 0}$ such that for all $i\in\mathbb{N}$, $i<|U|\Rightarrow W_i=U_i$, and $i\geq |U|\Rightarrow W_i=V_{i-|U|}$. Moreover, when it is clear from the context, we will write $U.x$ where $x$ is an element, here considered as a sequence of a single element.

A sequence of $k\in\mathbb{N} \cup \{\infty\}$ assemblies $\alpha_0,\alpha_1,\ldots$ over $\mathcal{A}^T$ is a \emph{$\mathcal{T}$-assembly sequence} if, for all $1 \leq i < k$, $\alpha_{i-1} \to_1^\mathcal{T} \alpha_{i}$. Alternatively, a $\mathcal{T}$-assembly sequence can be treated as a sequence of attachments $(t_i@z_i)_{i < k}$. The notation ``$t_i@z_i$'' means ``a tile of type $t_i$ attaches to a position $z_i$''.

The set of \emph{productions} of a tile assembly system $\mathcal{T}=(T,\sigma,\tau)$,
written $\prodasm{\mathcal{T}}$, is the set of all assemblies producible
from $\sigma$. An assembly $\alpha$ is called \emph{terminal} if there
is no $\beta$ such that $\alpha\rightarrow_1^{\mathcal{T}}\beta$. The
set of terminal assemblies is written $\termasm{\mathcal{T}}$.

For $A,B\in\mathbb{Z}^2$, we use the notation $\vect{AB}$ to mean ``the vector from $A$ to $B$''.

\subsection{Intrinsic simulations}

\label{sec:simulation_def}

These definitions are standard definitions of simulations in self-assembly. They are for a large part taken from~\cite{Meunier-2014}, and adapted to the space graph formalism used in this paper.

Let $T$ be a tile set, and let $m\in\N$. An \emph{$m$-block supertile} over $T$ is a partial function $\alpha : \Z_m^2 \dashrightarrow T$ (in the square grid), $\alpha:\mathbb{H}_m\dashrightarrow T$ (in the hexagonal grid), or $\alpha:\Z_m^3\dashrightarrow T$ (in the three-dimensional cubic grid), where $\Z_m = \{0,1,\ldots,m-1\}$ and $\mathbb{H}_m=\{a\vect{i}+b\vect{j}+c\vect{k} | a,b,c\in\{0,\ldots,m-1\}\}$.

Let $B^T_m$ be the set of all $m$-block supertiles over $T$.  The $m$-block with no domain is said to be $\emph{empty}$.  For a general assembly $\alpha:G \dashrightarrow T$ (where $G=\mathbb{Z}^2$ in the square grid, $G=\mathbb{H}$ in the hexagonal grid, and $G=\mathbb{Z}^3$ in the 3D cubic grid), and $\vect{v}\in G$, define $\alpha^m_{\vect{v}}$ to be the $m$-block supertile defined by $\alpha^m_{\vect{v}}(\vect{v'}) = \alpha(m\vect{v}+\vect{v'})$ for all $v'$ in $\Z_m^2$, $\mathbb{H}_m$ and $\Z_m^3$, respectively.

For some tile set $S$, a partial function $R: B^{S}_m \dashrightarrow T$ is said to be a \emph{valid $m$-block supertile representation} from $S$ to $T$ if for any $\alpha,\beta \in B^{S}_m$ such that $\alpha \sqsubseteq \beta$ and $\alpha \in \dom R$, then $R(\alpha) = R(\beta)$.

For a given valid $m$-block supertile representation function $R$ from tile set~$S$ to tile set $T$, define the \emph{assembly representation function}\footnote{Note that $R^*$ is a total function since every assembly of $S$ represents \emph{some} assembly of~$T$; the functions $R$ and $\alpha$ are partial to allow undefined points to represent empty space.}  $R^*: \mathcal{A}^{S} \rightarrow \mathcal{A}^T$ such that $R^*(\alpha') = \alpha$ if and only if $\alpha(A) = R\left(\alpha'^m_{A}\right)$ for all $A \in G$.  For an assembly $\alpha' \in \mathcal{A}^{S}$ such that $R(\alpha') = \alpha$, $\alpha'$ is said to map \emph{cleanly} to $\alpha \in \mathcal{A}^T$ under $R^*$ if for all non empty blocks $\alpha'^m_{A}$, $A+\vec{u} \in \dom \alpha$ for some $\vec{u}\in G$ such that $\|\vec{u}\|_1\leq 1$.

In other words, $\alpha'$ may have tiles on supertile blocks representing empty
space in $\alpha$, but only if that position is adjacent to a tile in $\alpha$.
We call such growth ``around the edges'' of $\alpha'$ \emph{fuzz} and thus
restrict it to be adjacent to only valid supertiles, but not diagonally adjacent
(i.e.\ we do not permit \emph{diagonal fuzz}).

By extension, the representation function from a cubic tileset $S$ to a planar square tileset $T$ is a regular $3D$ representation function where we consider the operations of tileset $T$ to happen in the $z=0$ plane.

In the following definitions, let $\mathcal{T} = \left(T,\sigma_T,\tau_T\right)$
be a tile assembly system, let $\mathcal{S} = \left(S,\sigma_S,\tau_S\right)$ be a
tile assembly system, and let $R$ be an $m$-block representation function $R:B^S_m
\rightarrow T$.

\begin{definition}
\label{def-equiv-prod} We say that $\mathcal{S}$ and $\mathcal{T}$ have
\emph{equivalent productions} (under $R$), and we write $\mathcal{S}
\Leftrightarrow \mathcal{T}$ if the following conditions hold:
\begin{enumerate}
\item $\left\{R^*(\alpha') | \alpha' \in \prodasm{\mathcal{S}}\right\} =
  \prodasm{\mathcal{T}}$.
\item $\left\{R^*(\alpha') | \alpha' \in \termasm{\mathcal{S}}\right\} = \termasm{\mathcal{T}}$.
\item For all $\alpha'\in \prodasm{\mathcal{S}}$, $\alpha'$ maps cleanly to $R^*(\alpha')$.
\end{enumerate}
\end{definition}

Another important condition that we require for simulation, is the equivalence of dynamics, via the two following definitions. The first one requires that whenever the simulator is capable of assembly a new supertile of some type $t$, the resulting production represents the assembly of a new tile in the simulated system:

\begin{definition}
\label{def-t-follows-s} We say that $\mathcal{T}$ \emph{follows}
$\mathcal{S}$ (under $R$), and we write $\mathcal{T} \dashv_R \mathcal{S}$ if
$\alpha' \rightarrow^\mathcal{S} \beta'$, for some $\alpha',\beta' \in
\prodasm{\mathcal{S}}$, implies that $R^*(\alpha') \to^\mathcal{T} R^*(\beta')$.
\end{definition}

The other direction is not as easy to define: intuitively, we require that whenever the simulated system is capable to assemble a new tile of type $t$, the simulator can also assemble a supertile representing $t$. However, because supertiles do not always determine which tile they represent by a single tile placement, but sometimes by many tile additions, the definition needs more care.

However, our results do not use the fully detailed version of this definition, instead requiring only that if a production $p$ of the simulated system is not terminal, then no production of the simulator representing $p$ is terminal.

\begin{definition}
\label{def-s-models-t}
We say that $\mathcal{S}$ \emph{models} $\mathcal{T}$ (under $R$), and we write
$\mathcal{S} \models_R \mathcal{T}$, if for every $\alpha \in
\prodasm{\mathcal{T}}$, there exists $\Pi \subset \prodasm{\mathcal{S}}$ where
$R^*(\alpha') = \alpha$ for all $\alpha' \in \Pi$, such that, for every $\beta
\in \prodasm{\mathcal{T}}$ where $\alpha \rightarrow^\mathcal{T} \beta$, (1) for
every $\alpha' \in \Pi$ there exists $\beta' \in \prodasm{\mathcal{S}}$ where
$R^*(\beta') = \beta$ and $\alpha' \rightarrow^\mathcal{S} \beta'$, and (2) for
every $\alpha'' \in \prodasm{\mathcal{S}}$ where $\alpha''
\rightarrow^\mathcal{S} \beta'$, $\beta' \in \prodasm{\mathcal{S}}$,
$R^*(\alpha'') = \alpha$, and $R^*(\beta') = \beta$, there exists $\alpha' \in
\Pi$ such that $\alpha' \rightarrow^\mathcal{S} \alpha''$.
\end{definition}

The previous definition essentially specifies that every time $\mathcal{S}$
simulates an assembly $\alpha \in \prodasm{\mathcal{T}}$, there must be at least
one valid growth path in $\mathcal{S}$ for each of the possible next steps that
$\mathcal{T}$ could make from $\alpha$ which results in an assembly in
$\mathcal{S}$ that maps to that next step.

\begin{definition}
\label{def-s-simulates-t} We say that $\mathcal{S}$ \emph{simulates}
$\mathcal{T}$ (under $R$) if $\mathcal{S} \Leftrightarrow_R \mathcal{T}$
(equivalent productions), $\mathcal{T} \dashv_R \mathcal{S}$ and $\mathcal{S}
\models_R \mathcal{T}$ (equivalent dynamics).
\end{definition}

\newcommand{\REPL}{\mathsf{REPR}} \newcommand{\frakC}{\mathfrak{C}}

\subsection{Intrinsic Universality}
\label{sec:iu_def}
Now that we have a formal definition of what it means for one tile system to
simulate another, we can proceed to formally define the concept of intrinsic
universality, i.e., when there is one general-purpose tile set that can be
appropriately programmed to simulate any other tile system from a specified
class of tile systems.

Let $\REPL$ denote the set of all supertile representation functions (i.e.,
$m$-block supertile representation functions for all $m\in\Z^+$).
Define $\frakC$ to be a class of tile assembly
systems, and let $U$ be a tileset.

\begin{definition}\label{def:iu-specific-temp}
We say $U$ is \emph{intrinsically universal} for $\frakC$ \emph{at temperature}
$\tau' \in \Z^+$ if there are functions $\mathcal{R}:\frakC \to
\REPL$ and $S:\frakC \to \mathcal{A}^U_{< \infty}$ such that, for each
$\mathcal{T} = (T,\sigma,\tau) \in \frakC$, there is a constant $m\in\N$ such
that, letting $R = \mathcal{R}(\mathcal{T})$,
$\sigma_\mathcal{T}=S(\mathcal{T})$, and $\mathcal{U}_\mathcal{T} =
(U,\sigma_\mathcal{T},\tau')$, $\mathcal{U}_\mathcal{T}$ simulates $\mathcal{T}$
at scale $m$ and using supertile representation function~$R$.
\end{definition}
That is, $\mathcal{R}(\mathcal{T})$ is a representation function that
interprets assemblies of $\mathcal{U}_\mathcal{T}$ as assemblies of
$\mathcal{T}$, and $S(\mathcal{T})$ is the seed assembly used to program
tiles from $U$ to represent the seed assembly of $\mathcal{T}$.

\begin{definition}
\label{def:iu-general}
We say that~$U$ is \emph{intrinsically universal} for $\frakC$ if it is
intrinsically universal for $\frakC$ at some temperature $\tau'\in Z^+$.
\end{definition}

\section{Simulations and the Bisimilarity Lemma}

\subsection{The Bisimilarity Lemma}

In order to look at local phenomena in the assembly, we define \emph{restricted dynamics}, where we are only concerned about what happens within a region of the space. This is the intuition behind the \emph{policy set} of a tile assembly system. Moreover, its \emph{diplomatic set} is the set of its possible communications with the rest of the space, through its border.

\begin{definition}[Glue movies]
  Let $\mathcal{T}=(T,\sigma,\tau)$ be a tile assembly system on some graph $G=(V,E)$.

  Let $F\subseteq E$ be a set of edges of $G$. A \emph{glue movie} on $F$ is a sequence $m$ of \emph{glue additions}, where for all $i\geq 0$, a glue addition $m_i$ is defined by $m_i = (e_i, s_i, g_i) \in F \times \{+,-\} \times \Glues(T)$. We call $e_i$ the edge of the glue addition, $s_i$ its orientation (i.e. on which side of the edge is the new tile attached), and $g_i$ its type.
Moreover, no two glue additions in a movie can have the same edge $e_i$; formally, for all $i,j\geq 0$ such that $i\neq j$, $e_i\neq e_j$.
\end{definition}

A glue movie is therefore a sequence of glue additions, oriented in a canonical way: for any $F\subseteq E$, an assembly sequence induces a (possibly empty) glue movie along the edges of $F$, in the following way: whenever a new tile $t$ attaches to the assembly using edges of $F$, we add to the glue movie all the edges of $F$ where $t$ matches the existing assembly, starting in clockwise order from the North of $t$. Moreover, the orientation $o\in\{+,-\}$ that we choose is $+$ for the south and west glues of $t$, and $-$ for its north and east glues (for the glues that are added to the glue movie).

\begin{definition}[Policy and diplomatic sets]
  Let $\mathcal{T}=(T,\sigma,\tau)$ be a tile assembly system on some graph $G=(V,E)$. Moreover, let $W\subseteq V$ be a set of vertices of $G$ with a local origin, and $F\subseteq E$ be the set of edges between $W$ and $V\setminus W$, with the same local origin.
  The \emph{policy set} $\mathcal{P}(\mathcal{T},W)$ of $\mathcal{T}$ is the set of the restrictions to vertices of $W$ of all possible assembly sequences of $\mathcal{T}$ (including all sequences not corresponding to a terminal assembly).
  The \emph{diplomatic set} $\mathcal{D}(\mathcal{T},F)$ of $\mathcal{T}$ along $F$ is the set of glue movies on $F$ defined by all possible assembly sequences of all productions of $\mathcal{T}$.
\end{definition}

\begin{definition}
   Let $\mathcal{T}=(T,\sigma,\tau)$ be a tile assembly system on some graph $G=(V,E)$. Let $X \subseteq E$ be a cut-set separating $G$ into two connected components $Y$ and $Z$ such that $\dom\sigma\subseteq Y$.

   Let $P = \mathcal{P}(\mathcal{T},Z)$, $D = \mathcal{D}(\mathcal{T},X)$, and let $P'$ be a set of prefixes of elements of $P$. The \emph{restriction of $D$ to $P'$}, noted $D_{|P'}$ is the set of all prefixes of $D$ where glues are only put next to tiles which are attached in $P'$.
\end{definition}

\begin{lemma}[Bisimilarity Lemma]
\label{lem:bisimilarity}
  Let $\mathcal{T}=(T,\sigma,\tau)$ be a tile assembly system on some graph $G=(V,E)$. There is a map $f$, from diplomatic sets to policy sets, such that for any cut-set $X\subseteq E$ separating $G$ into two connected components $Y$ and $Z$ such that $\dom\sigma\subseteq Y$, $\mathcal{P}(\mathcal{T},Z)=f(\mathcal{D}(\mathcal{T},X))$.

  (In other words, the policy set of a tile assembly system, in a zone $Z$ not containing any tile of $\sigma$, depends only on the diplomatic set on the border of $Z$).

  Moreover, $f$ is \emph{continuous}: for any such $X,Y,Z$, take $P'$ to be a set of prefixes of $\mathcal{P}(\mathcal{T},Z)$, $P'$ can be obtained as $f(D_{|P'})$.
\end{lemma}

Note that for any $X,Y$, $\mathcal{D}(\mathcal{T},X) = \mathcal{D}(\mathcal{T},X)$ implies that $Y$ is obtained by a translation of $Y$, as otherwise $\mathcal{D}(\mathcal{T},X)$ and $\mathcal{D}(\mathcal{T},X)$ are not even defined on the same set.

\subsection{Use of Lemma~\ref{lem:bisimilarity}}

Before we prove the Bisimilarity Lemma, let us review some of its uses and, for the reader familiar with the Window Movie Lemma \cite{Meunier-2014}, how it differs from that lemma.

Let us first consider the system $S$ defined on Figure~\ref{fig:bisim_example}. It is quite clear that its productions are all $1 \times n$ rectangles, with $n \operatorname{mod} 4 = 1$. We want to prove that it has a final production whose shape is a $1 \times n$ rectangle for arbitrarily large $n$. This can be proved ``by observation'', by exhibiting the relevant production for infinitely many values of $n$. But giving a proof through the Bisimilarity Lemma (Lemma \ref{lem:bisimilarity}) will provide insight into how to apply the lemma to cases where observation of the tileset in infeasible.

\begin{figure}
  \centering
  \includegraphics[width=\textwidth]{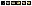}
  \caption{A simplistic tile assembly system}
  \label{fig:bisim_example}
\end{figure}

First, let us notice that no production $S$ has a tile outside of the $x = 0$ line. Let us consider for all $i > 0$ the zone $Z_i = \{(x,y) | x > i, y \in \Z\}$. For any $i$, $\mathcal{D}(S, Z_i)$ is a set of glue movies on the set of edges \( \partial Z_i = \{ ((i,y),(i+1,y)) | y \in \Z\} \). These glue movies have at most two glue additions, at $\{(i,0),(i+1,0)\}+$ and $\{(i,0),(i+1,0)\}+$. There are $6$ different glues in $S$, therefore there are at most $36 + 6 + 1 = 43$ possible glue movies of length $0$, $1$ or $2$\footnote{Actually, the glue movie is determined by the first glue, so there are only $7$ possible movies.}. Thus, there are at most $2^{43}$ different values of $\mathcal{D}(S, Z_i)$.

Since $S$ has a production of length greater than $2^{43}$ (this can be seen by actually building such a production), there are $i < j$ such that $\mathcal{D}(S, Z_i) = \mathcal{D}(S,Z_j)$, and $\mathcal{D}(S, Z_i)$ is not reduced to the empty glue movie. By the Bisimilarity Lemma, $\mathcal{P}(S,Z_j) = \mathcal{P}(S,Z_i)$. Because of this, for any production $p$ of $S$, there is a production $p'$ such that $p' \cup Z_j$ is $(p \cup Z_i)$ translated by $(j-i,0)$. But then, $p'$ is $j-i$ tiles longer than $p$, from which we get the result.

This proof shows how to ``pump'' an assembly using the Bisimilarity Lemma. In this example, the same result can be achieved by the Window Movie Lemma. 

The proof of Theorem~\ref{thm:nomismatches} shows how the two lemmas differ.

\subsection{Proof of Lemma~\ref{lem:bisimilarity}}
\begin{proof}

  We first define $f$ for singletons: let $S$ be a glue movie on $X$. We build a set $f(\{S\}) \subseteq (Z\times T)^{\mathbb{N}}$ of assembly sequences of $\mathcal{T}$, by induction on the length of $S$.

  First let $f_0=\{\emptyset\}$. Then, for all $i\geq 0$, let $(e_i,s_i,g_i)=S_i$ be tile attachment of index $i$ in the assembly sequence.
  \begin{itemize}
  \item If $e_i=(a,b)$, $b\in Z$ (\resp $a\in Z$) and $s_i=+$ (\resp $s_i=-$), let $T_i$ be the set of assembly sequences not involving any glue of $X$, that can grow from some assembly induced by an assembly sequence of $f_i$, and that is compatible with all the glue additions of $S$ involving vertex $b$ (\resp vertex $a$).

    Remark that there are at most four such glue additions in $S$, and that, by the construction of an induced movie, they are all consecutive in $S$. Also, by this definition, any prefix of an assembly sequence of $T_i$ is also in $T_i$.

    Then, let $S_{i+1}=S_i\cup T_i$.
  \item Else, let $S_{i+1}=S_i$.
  \end{itemize}

If $S$ is finite, of length $n$, this construction stops after $m\leq n$ steps, and we let $f(\{S\})=S_m$. Else, this construction never stops. Let $f_{\infty}=\bigcup_{i\in \mathbb{N}} f_i$, and $f(\{S\})$ be the union of $f_{\infty}$ and the set of all limits of increasing sequences (with respect to the prefix order) of words from $f_{\infty}$.
Moreover, for any diplomatic set $D$, $f(D)$ is defined by $f(D)=\bigcup_{m\in D} f(\{m\})$.

Finally, we prove that our construction for $f$ satisfies our claims:

\paragraph{\bf $f$ is complete, i.e. $\mathcal{P}(\mathcal{T},Z)\subseteq f(\mathcal{D}(\mathcal{T},X))$.}

Let $\alpha$ be a finite production of $\prodasm{\mathcal{T}}$, and $S_\alpha$ be an assembly sequence of $\alpha$. By an immediate induction on the length of $S_\alpha$, its restriction to $Z$ is in $f(\mathcal{D}(\mathcal{T},X))$.

Now, if $\alpha$ is an infinite production of $\prodasm{\mathcal{T}}$, any assembly sequence $S_\alpha$ of $\alpha$ is also in $f(\mathcal{D}(\mathcal{T},X))$: indeed, let $M_\alpha$ its induced glue movie on $Z$. Either $M_\alpha$ uses the edges of $X$ only a finite number of times, and therefore one of its suffixes is included in $T_i$ for some $i$, or it uses the edges of $X$ an infinite number of times, and then the construction of $f(\{S_\alpha\})$ has infinitely many steps, each with a finite number of tile additions.

Therefore, in this case, all finite prefixes of $S_\alpha$ are in $f(\mathcal{D}(\mathcal{T},X)$. Since $S_\alpha$ itself is the limit of these prefixes, it is therefore also in $f(\mathcal{D}(\mathcal{T},X))$.

\paragraph{\bf $f$ is sound, i.e. $f(\mathcal{D}(\mathcal{T},X))\subseteq \mathcal{P}(\mathcal{T},Z)$.}

We now prove that for any assembly sequence $s\in f(\mathcal{D}(\mathcal{T},X))$, there is a valid assembly sequence $s^\beta$ of some $\beta\in\prodasm{\mathcal{T}}$ such that $s$ is the restriction of $s^\beta$ to $Z$.

Since $s\in f(\mathcal{D}(\mathcal{T},X))$, there is a glue movie $m\in\mathcal{D}(\mathcal{T},X)$ on $X$ such that $s\in f(\{m\})$. By definition of the diplomatic set, $m$ is the induced movie of some assembly sequence $r$ of $\mathcal{T}$.

Now, we build $s^\beta$ by induction: we will build a sequence $(s^\beta_n)_{n\in\mathbb{N}}$ of assembly sequences, placing at least one tile at each step, with the invariants that for all $n\geq 0$, $s^\beta_n$ is a valid assembly sequence of $\mathcal{T}$, and the restriction of $s^\beta_n$ to $Z$ is a prefix of $s$.

We need a few auxiliary sequences, also defined inductively: $(a_n)_{n\in\mathbb{N}}$ is the sequence of indices in $r$, starting with $a_0$, defined as the smallest integer such that $r_{a_0}$ is placed in $Y$, and $(u_n)_{n\in\mathbb{N}}$ is the sequence of indices in $s$, starting at $u_0=0$.

\begin{itemize}
  \item If possible, let $b_n\geq a_n$ be the largest integer such that all tile attachments of $r_{a_n,a_n+1,\ldots,b_n}$ be in $Y$, and $s_n=s^\beta_n\cdot r_{a_n,a_n+1,\ldots,b_n}$ be a valid assembly sequence.

    In this case, let $s^\beta_{n+1}=s_n$, and let $a_{n+1}>b_n$ be the smallest integer such that $r_{a_{n+1}}$ is a tile attachment in $Y$. Also, let $u_{n+1}=u_n$.

  \item else, if movie $m$ is infinite, or if $i<|m|$, we can move forward in $m$: indeed, the next step of $r$ (at $b_n+1$) necessarily involves a glue addition on some edge of $X$. Moreover, since $m$ is the movie from which we built $s$, tile attachment $s_{u_{n}}$ can be done in $\beta^n$. Let thus $s^\beta_{n+1}=s^\beta_{n}\cdot s_{u_{n+1}}$, and let $u_{n+1}=u_n+1$. Also, let $a_{n+1}=a_n$.

\end{itemize}
\end{proof}

\subsection{Differences with known results}

Remark how this definition and lemma differs from the window movie lemma~\cite{Meunier-2014}: by Lemma~\ref{lem:bisimilarity}, if two diplomatic sets are the same up to translation, we can ``swap'' their whole policy sets (formally, the same tile assembly system can also produce an assembly with the part between the windows taken out).
In other words, we can move partial assemblies back and forth between these two translations of the window, at \emph{any} step of the dynamics.

In contrast, the window movie lemma uses a weaker hypothesis: in our formalism, it just asks for two diplomatic sets to intersect. It also yields a weaker conclusion, by showing that we can ``swap'' translations of partial assemblies \emph{only once}.

\section{Systems without mismatches do not simulate general systems}

This section proves Theorem~\ref{thm:nomismatches}; the proof is split into three subsections: Section~\ref{preserve} shows that we can switch back and forth between the cutspace and $\mathbb{Z}^2$, without changing geometric properties of the simulation. Section~\ref{keepgoing} shows a slightly stronger statement, namely that the dynamics are also the same in $\mathbb{Z}^2$ and in the cutspace. Altogether, these section intuitively show that the tileset cannot ``detect'' that it is being ran in the cutspace. Finally, Section~\ref{unify} shows how to reconcile sequences that happen in the cutspace into a valid sequence in $\mathbb{Z}^2$, that a claimed simulator without mismatches must be able to grow.

We say that two adjacent tiles \emph{do not match} if at least one of
them has a strictly positive glue on their common side, and their respective
colors on that side are different.

\begin{reptheorem}{thm:nomismatches}
There is a tile assembly system $\mathcal{T}=(T,\sigma,1)$, such that for all simulator $\mathcal{U}$ of $\mathcal{T}$, there are assemblies $\alpha\in\prodasm{\mathcal{U}}$ where at least two adjacent tiles of $\alpha$ do not match on their abutting side.
\end{reptheorem}
\begin{proof}
Let $\mathcal{T}=(T,\sigma,1)$ be the tile assembly system with the tiles of Figure \ref{fig:cutspace-tas}, and $\sigma$ is the bottom-left tile of Figure \ref{fig:cutspace-tas}, at position $(0,0)$.

  \begin{figure}[ht]
      \includegraphics[width=\textwidth]{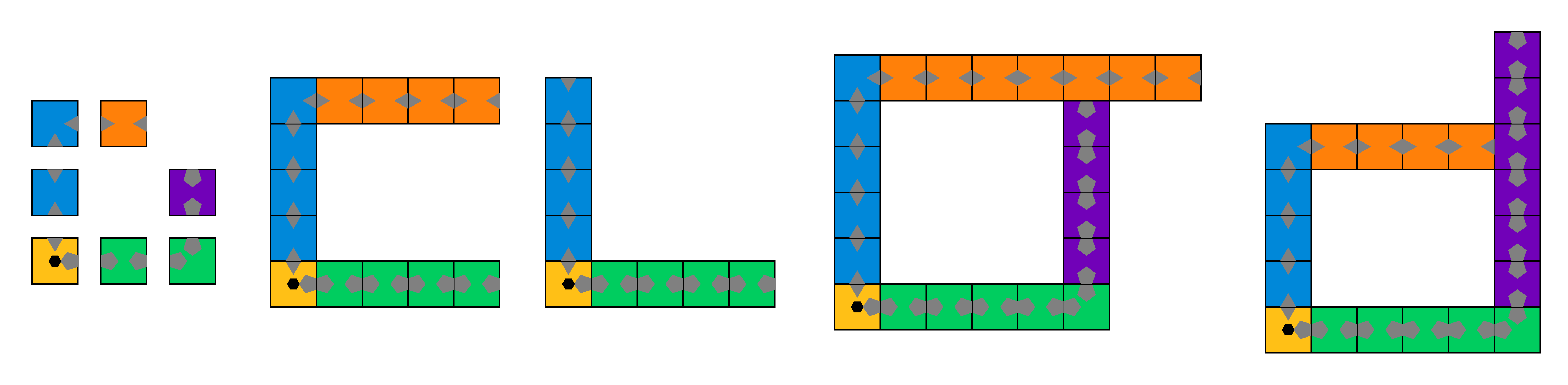}
    \caption{A tile assembly system at temperature 1; the seed tile is marked with a hexagon; in the pictured productions, the \emph{upper arm} is in dark, and the \emph{lower arm} in light.}
    \label{fig:cutspace-tas}
  \end{figure}

One can readily see that the productions of $\mathcal{T}$ are all one of the shapes shown on Figure \ref{fig:cutspace-tas}. More precisely, they each have one \emph{upper} and one \emph{lower} arm (as pictured on Figure \ref{fig:cutspace-tas}); if the bottom arm turns left and grows upwards, and the upper arm turns right and grows to the right, then the two arms compete at the intersection. The first to arrive continues towards infinity, and the other ``crashes'' into it.

\begin{figure}[h!]
  \includegraphics[width=\textwidth]{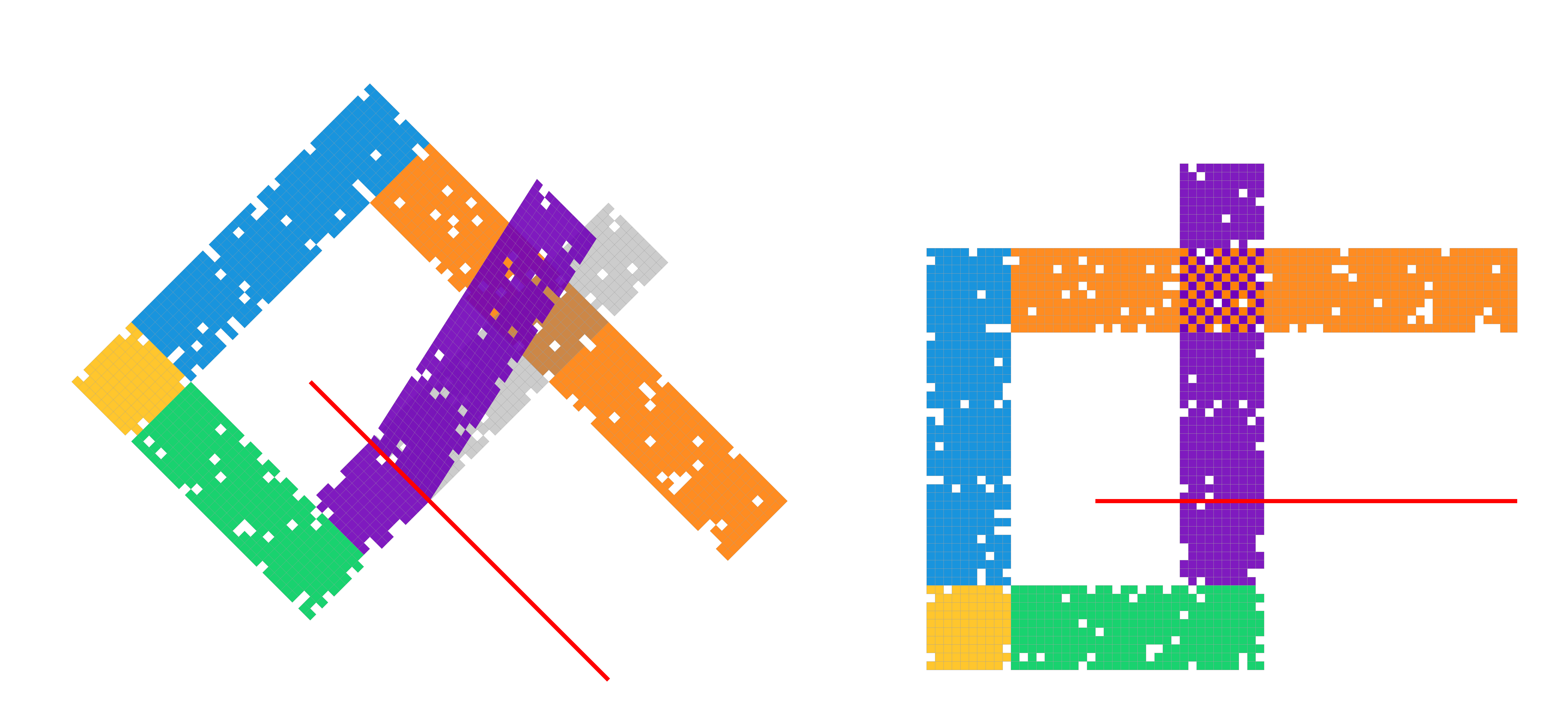}
  \caption{Left: Conflict between the two arms is averted in the
    cutspace: the upper arm passes ``over'' the lower arm. Right:
    trying to embed this production in $\Z^2$ might lead to a
    conflict in the intersection zone unless the tiles there coincide.}
  \label{fig:cutspace-compatible}
\end{figure}

Now, assume for the sake of contradiction, that there is a tile assembly system $\mathcal{U}$, that simulates $\mathcal{T}$ with scale factor $s$, and that $\mathcal{U}$ does so without mismatches. We will show that $\mathcal{U}$ is also able to assemble a production where the two arms cross, and both grow infinitely far. Since this production does not represent any production of $\prodasm{\mathcal{T}}$, this will contradict our hypothesis that $\mathcal{U}$ can simulate $\mathcal{T}$.

\subsection{The cut-space, and how to embed assemblies of $\mathbb{Z}^2$ into it, that preserve diplomatic and policy sets}
\label{preserve}

First, let us define an auxiliary assembly grid $G=(V,E)$, called the \emph{cut-space}, where the set $V$ of vertices is made of two copies of $\mathbb{Z}^2$: $V=V_0\cup V_1$, where $V_0=\{(x,y,0) | (x,y)\in\mathbb{Z}^2\}$ and $V_1=\{(x,y,1) | (x,y)\in\mathbb{Z}^2\}$.
Let $E_0$ and $E_1$ be the set of edges within each copy of $\Z^{2}$, and for $i\in\{0,1\}$, let $E_i=\{((x,y,i),(x',y',i)) | x,y,x',y'\in V\hbox{ and }|x-x'|+|y-y'|=1)\}$.

The set $E$ of edges of the cut-space is defined by exchanging the ends of the edges along a ray defined in $\mathbb{R}^2$ by $y=2s+0.5$ and $x\geq 2$. More precisely:
\begin{eqnarray*}
E&=&(E_1\cup E_0)\\
&&\setminus\left\{((x,2s,i),(x,2s+1,i)) | x\geq 2,i\in\{0,1\}\right\}\\
&&\cup\left\{((x,2s,i),(x,2s+1,1-i)) | x\geq 2,i\in\{0,1\}\right\}
\end{eqnarray*}

The \emph{projection} $\pi$ of the cutspace $G$ into $\mathbb{Z}^2$ is the map defined on $G$ by $\pi(x,y,z)=(x,y)$.

All positions in the cutspace still have four neighbors. Remark that the geometry of this graph makes the intuitive model of assemblies with square tiles less intuitive, as two surfaces could be assembled, that cross each other.
However, in our construction, we will only use tiles in one direction at these positions, the left arm growing upwards. The idea of defining this space is to grow both assemblies of Figure \ref{fig:cutspace-tas} in the same space, as shown on Figure \ref{fig:cutspace-compatible}.

For each $y \geq 3$, the set $C_y = \{((x,y,1),(x,y+1,1)) | x \in \Z\}$
is a cut-set of the cutspace. Also, for each $x\geq 3$,
$D_x = \{((x,y,0),(x,y+1,0)) | y>3 \}\cup \{((x',y,0),(x'+1,y,0)) | x'\geq x \}$
is a cut-set of the cutspace. Moreover, $E_k=C_k\cup D_k$ is also a cut-set of $G$, separating it into three parts, shown on Figure \ref{fig:cutset}. Likewise, define $A_{>y} = \{(x,y',1) | x \in \Z, y' > y\}$, $B_{>x} = \{(x',y,0) | y>3 \wedge x' > x\} \}$ and $V_{>k} = A_{>k} \cup B_{>k}$ to be the ``far side'' of $C_x$, $D_y$ and $E_k$. We will use these parts to describe the growth of the arms ``past the crossing''.

Now, let $W,H$ be two integers. Let $p$ be a production of $\mathcal{T}$ in $\Z^2$ in which the lower arm has run $W$ tiles to the right and one tile upwards, and the upper arm has run $H$ tiles upwards and one tile to the right. We know that continuing to run $\mathcal{T}$ from $p$ is going to enclose a $W \times H$ rectangle.

Since $\mathcal{U}$ simulates $\mathcal{T}$ with scale factor $s$, there is a production $q$ of $\mathcal{U}$ which is a representation of $p$. From now on, we will study productions which can be obtained from $q$. Define an auxiliary tile assembly system $\mathcal{U}' = (U, q, \tau)$, which is $\mathcal{U}$ with $q$ as the new seed.
Remark that any production of $\mathcal{U'}$ is also producible by $\mathcal{U}$.

We do not know what is the behaviour of  $\mathcal{U}'$ in the cutspace; it does not follow from the definition of simulation that it should behave like $\mathcal{T}$ does in the cutspace. 

Let $\epsilon$ be an embedding, mapping positions $z=(x,y)$ with $y \geq 2$ and $x \geq 2$ to $\epsilon(z) = (x,y,1)$, and the other positions to $\epsilon(z) = (x,y,0)$. $\epsilon$ maps adjacent positions in $\Z^2$ to adjacent positions in the cutspace.

Let $u = (t_i@z_i)_{i\in\N}$ be an assembly sequence of $\mathcal{U}'$ in $\Z^2$. Define
$\epsilon(u)$ \emph{in the cutspace} as the embedding of this sequence, i.e. the sequence defined by $(t_0@\epsilon(z_0)), (t_1@\epsilon(z_1)),\ldots,t_n@\epsilon(z_n)$.  Observe that $\epsilon(u)$ is an assembly sequence of $\mathcal{U}'$ in the cutspace, since adjacency is preserved under $\epsilon$. Its upper arm lives in the $0$ layer of the cutspace, and its lower arm in the $0$ layer for its horizontal part, and in the $1$ layer for its upwards part. Therefore, any assembly sequence of $\mathcal{U}'$ in $\Z^2$ yields by $\epsilon$ an assembly sequence of $\mathcal{U}'$ in the cut-space where the arms are at least as long.

\begin{figure}[ht]
\centering

\includegraphics[scale=0.3]{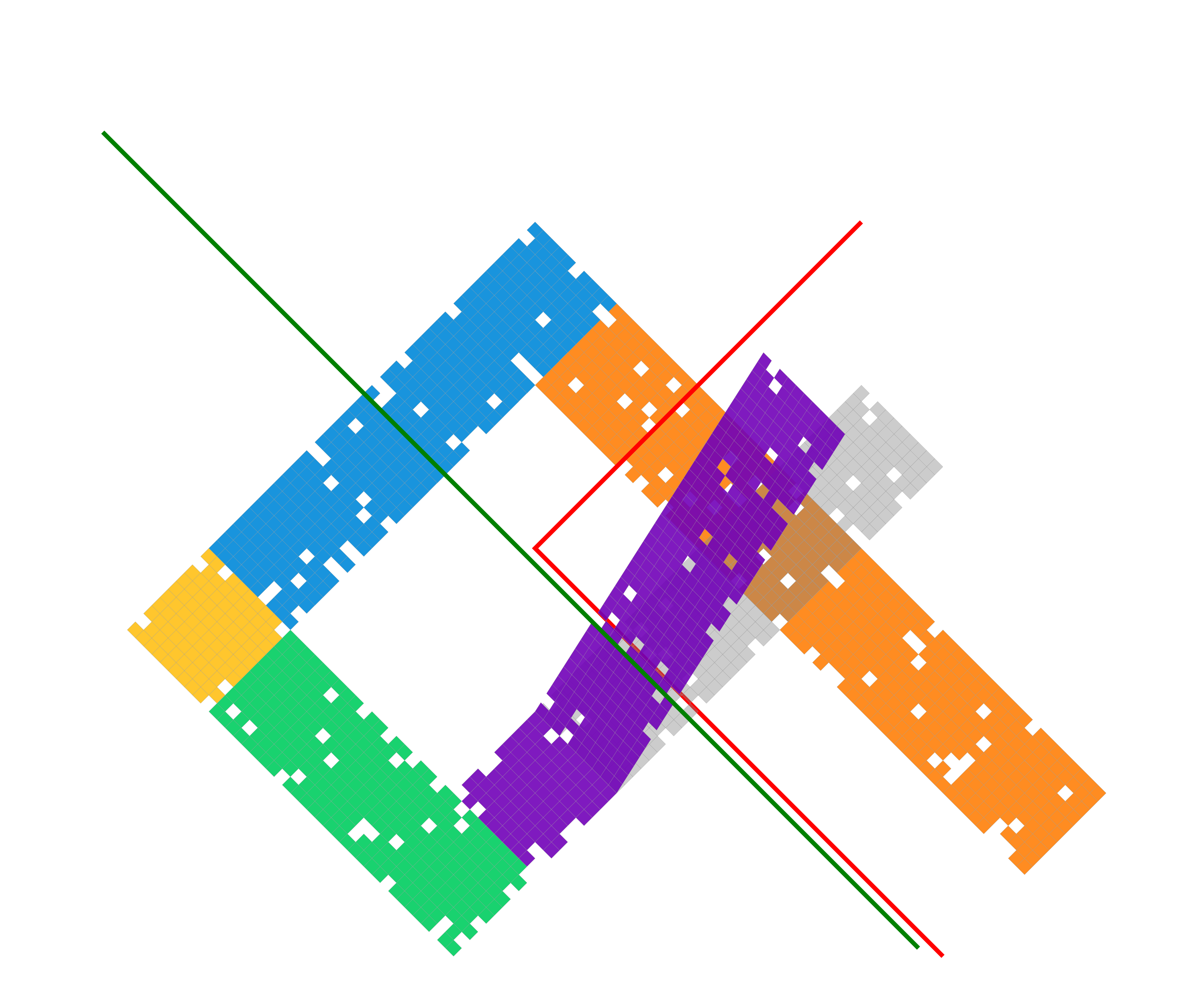}
\caption{$E_k$ partitions the cut-space into three parts. On this figure, $C_k$ is in green, and $D_k$ is in red. Note that $C_k$, does \emph{not} intersect with the blue part of the assembly (upper left), but floats over it. Likewise, $D_k$ does not intersect the purple arm.}
\label{fig:cutset}
\end{figure}

Conversely, take an assembly sequence $v$ of $\mathcal{U}'$ in the cutspace which produces $p$ where either the upper arm has length less than $sH$, or the lower arm has length less than $sW$. Then for all $z,z' \in p$ such that $z\neq z'$, $\pi(z) \neq \pi(z')$. Therefore, by mapping $\pi$ over $v$, one gets a configuration $\pi(p)$ which is a production of $\mathcal{U}'$ on $\Z^2$. Thus, if an assembly sequence of $\mathcal{U}'$ in the cutspace places a tile further away than $s$ tiles from the arms, it does so after both arms have passed position $(sH,sW,l)$, where $l$ is their respective layer. On the other hand, if $H$ and $W$ are large enough, namely, if
\(N=\min(H,W) > 2^{(4|U|+1)^{6s}\cdot (6s)!}\), then there must be \(k, k' < N\) such that the prefixes of the diplomatic sets on $E_k$ and $E_k'$ before any tile can be put further away than $s$ from the arms are the same, since all prefixes of movies before that event have at most a finite number of tile addition. By the laziness property of Lemma~\ref{lem:bisimilarity}, the prefix of the policy sets of $V_{>k}$ and $V_{>k'}$ before any tile placement outside of the arms must be the same. But then, a minimal assembly sequence which adds a tile outside the arms must have each arm have length less than $k'$, which leads to a contradiction. Therefore, $\mathcal{U'}$ only places tiles within $s$ tiles of the position of the arms on the cutspace, just as it does on $\Z^2$.

{\bf Conclusion of this section:} for large enough $W$ and $H$, there are two indices $k,k'$ such that $0 < k < k' < \min(W,H)$ such that the full diplomatic sets $\mathcal{D}(\mathcal{U}',E_k)$ and $\mathcal{D}(\mathcal{U}',E_{k'})$ (in the cut space) are the same, up to a translation by a  vector $(k'-k,0,0)$ for $D_k$, and by a vector $(0,k'-k,0)$ for $C_k$ (which is possible because $D_k$ and $C_k$ are disjoint).
Moreover, if $W>4Ns$ and $H>4Ns$, $k$ and $k'$ can be taken such that $k' - k > 3s$.

\subsection{Getting the arms to cross in the cut space}
\label{keepgoing}
The first lemma we will need is essential to prove that we can ``keep the dynamics going'', i.e. that we can add tiles, in the cut space, until the two arms ``cross'' (of course, since they are not on the same plane, the bottom arm ``jumps over'' the top arm).

One issue is, although we are simulating a system in $\mathbb{Z}^2$ which, if ran in the cut space, would produce two infinite arms nicely crossing each other, the simulation of this system in $\mathbb{Z}^2$ could send signals between the arms to detect that it is being ran in the cut space, and stop the growth of one of the arms, allowing a successful simulation.

However, the following lemma shows that no tileset can use such a strategy. To show this, we use once again the bisimilarity lemma:

\begin{lemma}
\label{lem:crocut}
Let $u$ be an assembly sequence of $\mathcal{U}'$ in the cutspace, which produces some assembly $p$. If $p$ does not have at least one tile at $(s(H+1)+1,y_H,0)$ (i.e. strictly to the right of the bottom arm) for some $y_H$ (respectively, $(x_W, s(W+1)+1,1)$ for some $x_W$, i.e. strictly above the top arm), then $p$ is not terminal. Moreover, in any attachment sequence $u'$ having $u$ as a prefix, there is an integer $i > |u|$ such that $u'_i$ is an attachment on the lower arm (respectively, the upper arm).
\end{lemma}

\begin{proof}
  Let $u$ be an assembly sequence ending in a production $p$ which does not have a tile at  $(s(H+1)+1,y_H,0)$ nor at $(x_W, s(W+1)+1,1)$ for any $y_H$ and $x_W$. By the Bisimilarity Lemma applied to $E_{k'}$ and $E_k$, one can get a corresponding assembly sequence $u'$ which does not have a tile at $(s(H-2)+1,y_H,0)$ nor at $(x_W, s(W-2)+1,1)$ for any $y_H$ and $x_W$ (see Figure~\ref{fig:diplo1}). But $\pi(u')$ is a valid assembly sequence in $\Z^2$, and in any of its prolongations, both arms must get at least $s$ rows or columns past $(Hs,Ws)$.

Therefore, in the policy set of $E_{>k}$, there must be a way to extend $u$ further than $s(H-1)$ to the right on layer $0$ or without tiles further upwards than $s(W-1)$ on layer $1$. Thus, by translating by $(k-k')$ to get the policy set of $E_{>k'}$, we get the expected result.
\end{proof}

Note that this lemma states a property relating \emph{all} assembly sequences of $\mathcal{U}'$ in two translated regions. Hence, its proof needs the full Bisimilarity Lemma (except the continuity of $f$): indeed, the Window Movie Lemma would allow us to work only on \emph{one} assembly sequence.

\subsection{Unifying assembly sequences of crossing arms, in the projection of the cut-space onto $\mathbb{Z}^2$}
\label{unify}

\begin{definition}
  Let $u$ be an sequence of attachments of $\mathcal{U}'$ in the cutspace. $u$ is \emph{consistent} if for any $t@z, t'@z' \in u$, if $\pi(z) = \pi(z')$, then $z = z'$

  For a consistent sequence of attachments $u$, we define $\Pi(u)$ to be the subsequence of first occurrences of elements in $(\pi(x))_{x \in u}$. In other words for each positions whose projection appears twice in $u$, $\Pi$ only keeps the first occurrence.
\end{definition}

\begin{lemma}
  \label{lem:reconcile}
  Let $u$ be a consistent assembly sequence of $\mathcal{U}'$ in the cutspace, and let $p$ be the production assembled by $u$. Let $i < |u'|$, and $t@(x,y,l) = u'_i$ such that $y > 2s + 1$. Then, if $(x,y,1-l)$ is attachable in $p$, $t@(x,y,1-l)$ is a valid attachment onto $p$, and thus, $u . t@(x,y,1-l)$ is a consistent assembly sequence.
\end{lemma}

\begin{figure}[ht]
\centering
\includegraphics{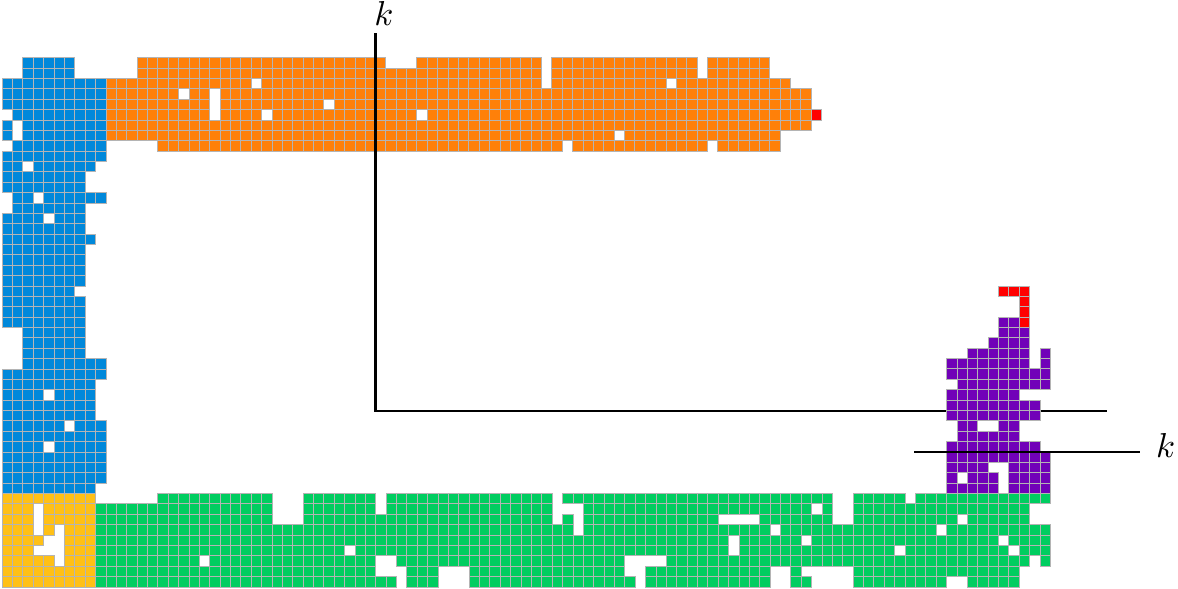}
\caption{A production (in the cutspace) where both arms are short must be able to place at least one tile on each arm. These tiles are drawn in red.}
\label{fig:diplo1}
\end{figure}

\begin{proof}
  Since $u$ is a consistent sequence, $\Pi(u)$ is an assembly sequence of $\mathcal{U}'$ in $\Z^2$, which yields $\pi(p)$ as production. Since $(x,y,1-l)$ is attachable and $y > 2s+1$, there are some neighbors of $(x,y,1-l)$ in $p$ with total glue abutting $(x,y,1-l)$ at least $\tau$, and they are on layer $1-l$. These neighbors map by $\pi$ onto neighbors of $(x,y)$, and since $\pi(p)$ is without mismatches, the glues they have on their sides abutting with $(x,y)$ match those of $t$. Therefore, $t@(x,y,1-l)$ is a possible attachment in $p$.
\end{proof}

We finish the proof by iteratively reconciling assembly sequences until the two arms cross in $\mathbb{Z}^2$, yielding the desired contradiction:

\begin{lemma}
  \label{lem:cross}
  There is an integer $n_0$ and a sequence $(u^n)_{n\leq n_0}$ of
  consistent assembly sequences in the cutspace such that for some $y_h, x_V, l, l'$, $u^{n_0}$ has an attachment at $((s+1)H+1,y_H,l)$
  and one at $(x_V, (s+1)V+1,l')$. Moreover, for $1 \leq n \leq n_0$,
  $|u^n| > |u^{n-1}|$.
\end{lemma}

With $u^n$ given by Lemma \ref{lem:cross}, $\Pi(u^{n_0})$ is an assembly sequence of $\mathcal{U}'$ in $Z^2$ with the two arms crossing. Thus, $\mathcal{U}'$ has a production which does not represent any production of $\mathcal{T}$. From this, it follows that $\mathcal{U}$ does not simulate $\mathcal{T}$.

\begin{proof}
  Let $u^0$ be an empty sequence.
  For all $n\geq 0$, $u^{n+1}$ is defined as follows:

\begin{description}
\item[Success.] If $\Pi(u^n)$ is a consistent assembly sequence of $\mathcal{U}'$ with the two arms crossing (i.e., with tiles at positions $((s+1)H+1,y_H, l)$ and $(x_V, (s+1)V+1, l')$ for some $y_H$ and $x_V$), then $n_0 = n$.
\item[Non-conflicting case.] Else, $u^n$ does not represent a terminal assembly because of Lemma \ref{lem:crocut}. Therefore, a new attachment $t@z$ is possible after $u^n$ in the arm that is too short. If $u^n . t@z$ is a consistent sequence, then let $u^{n+1} = u^n. t@z$
\item[Conflicting case.] Lastly, if $u^n . t@z$ is not consistent, then by Lemma \ref{lem:reconcile}, there is a $t'$ such that $u^n. t'@z$ is consistent. In that case, let $u^{n+1} = u^n . t'@z$.
\end{description}

Condition ``Success'' must be reached after a finite number of stpes, since only a finite number of attachments can be made in each arm before a tile is placed at both $((s+1)H+1,y_H, l)$ and $(x_V, (s+1)V+1, l')$. This proves that the desired sequence does exist and concludes the proof.

\begin{figure}[ht]
\centering
\includegraphics{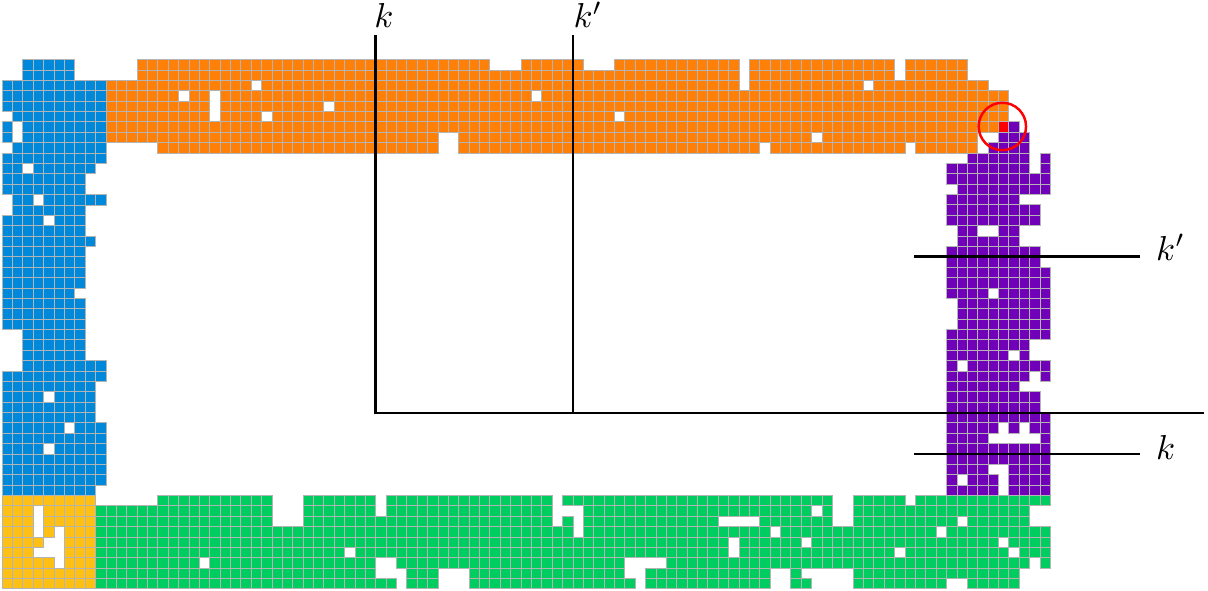}
\caption{The first conflict in the projection of the cut-space into $\mathbb{Z}^2$ is on the tile circled in red.}
\label{fig:diplo0}
\end{figure}
\end{proof}

\end{proof}

\section{Local consistency is strictly stronger than the mismatch-free property}

Local consistency is the property of a tile assembly system that no production has mismatches, and that all tiles are always placed such that there is no \emph{excess of glue}, i.e. such that the sum of their glue strengths is \emph{exactly} equal to the temperature.
This property has been used in \cite{USA} to construct an intrinsically universal tileset for a restricted class.

However, the classical mismatch-free constructions of Turing machine simulations~\cite{RotWin00,SolWin05} are also locally consistent. Moreover, the recent separation result of temperatures 1 and 2~\cite{Meunier-2014}, also applies to separate locally consistent temperature 2 systems from temperature 1. This raises the question of whether local consistency is stronger than the simpler mismatch-free property, which we answer with the following theorem:

\begin{reptheorem}{thm:local-vs-mismatch-free}
  There is a mismatch-free tile assembly system $\mathcal{T}=(T,\sigma,1)$, such that no tile assembly system $\mathcal{U}$ (of any temperature) simulating $\mathcal{T}$ is locally consistent.
\end{reptheorem}
\begin{proof}
    The idea is similar to the proof of Theorem \ref{thm:nomismatches}: we will exhibit a very simple system that does not have mismatches but is not locally consistent, and then repeatedly use Lemma \ref{lem:bisimilarity} on a claimed simulator $\mathcal{U}$ of $\mathcal{T}$ to show that at least one tile, in a production of $\mathcal{U}$, must be placed with too much glue.

    In this proof, $\mathcal{T}$ will be the temperature 1 system with one uniform tile, one copy of which is its seed tile. $\mathcal{T}$ is described on Figure~\ref{fig:locally}.

    \begin{figure}[ht]
      \centering
      \includegraphics[scale=0.15]{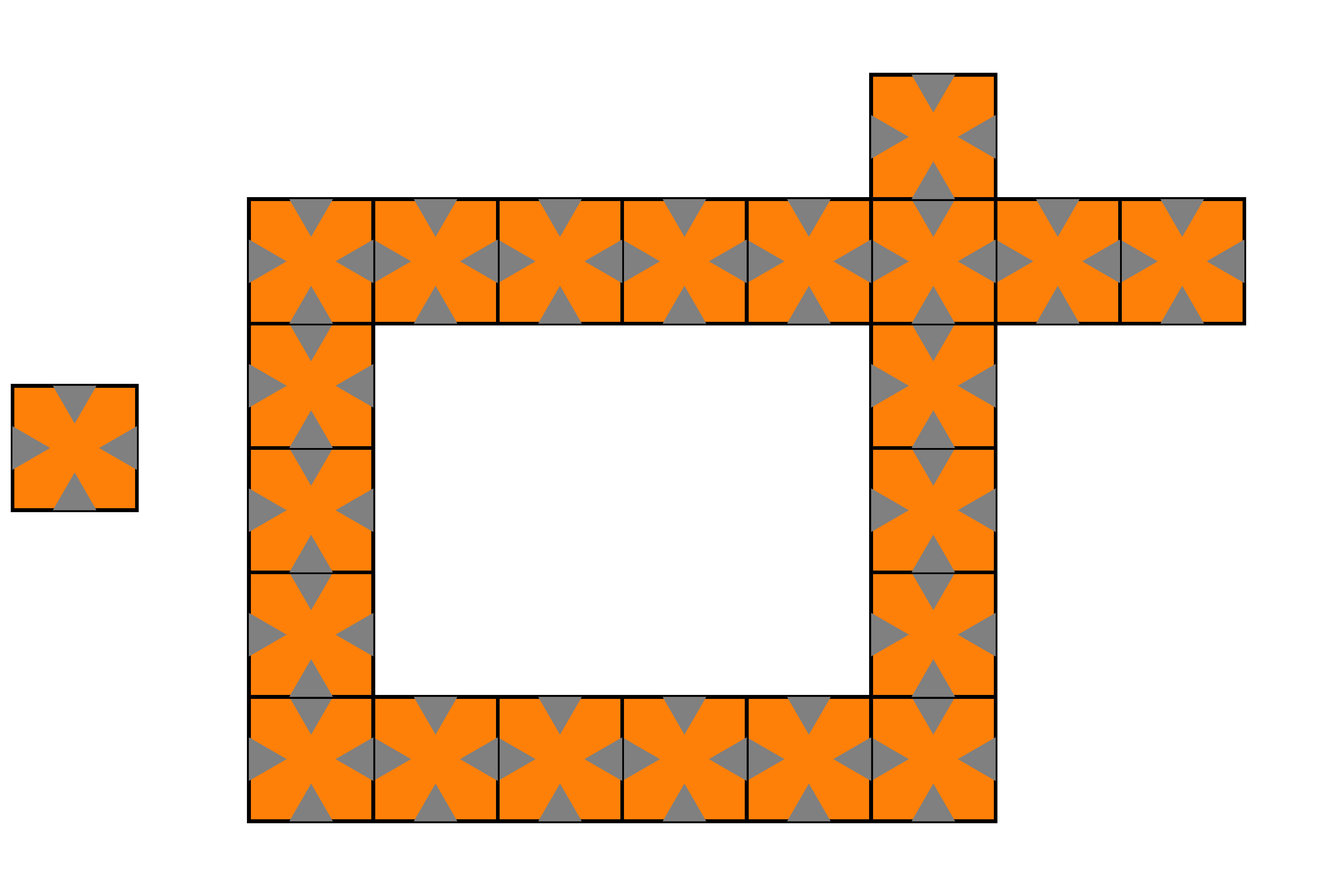}
      \caption{Tile assembly system $\mathcal{T}$, with an example production. All glues have strength $1$. Remark that the tiles at the junctions of two branches can always be placed with an excess of glue.}
      \label{fig:locally}
    \end{figure}

    We will show that any simulator $\mathcal{U}$ for $\mathcal{T}$, at some scaling factor $f$, has either a mismatch or an excess of glue, for at least one tile attachment.
    First, $\mathcal{T}$ can produce an assembly $\alpha$ with two parallel arms, a bottom arm and a top arm, both of length at least $N$, where $N$ is the constant defined as in the proof of Theorem~\ref{thm:nomismatches}, and where the bottom arm turns left after $N$ steps (i.e. a prefix of the construction shown on Figure~\ref{fig:locally}).

    By hypothesis, $\mathcal{U}$ must be able to produce an assembly $\beta$ representing $\alpha$ at scaling factor $f$. Let $\beta_0$ be the largest prefix of $\beta$ such that the two arms do not cross, and let $\gamma_0$ be the assembly equivalent to $\beta_0$ in the cut-space defined in the proof of Theorem~\ref{thm:nomismatches}.
    By Lemma~\ref{lem:bisimilarity}, there are two integers $k$ and $k'$, with $k'\geq k+10 f$ such that diplomatic sets $\mathcal{D}(\mathcal{U},E_k)$ and $\mathcal{D}(\mathcal{U},E_{k'})$ are translations of each other. We can therefore get an assembly sequence $\delta$ in the cutspace where both arms are ``long'', i.e. pass position $f.N$ in their respective layer, by the same argument as in Lemma~\ref{lem:crocut} (namely, because we can find an assembly with ``short'' versions of the arms, that can place at least one more tile because such an assembly cannot represent a terminal assembly).

    \begin{figure}[ht]
      \centering
      \includegraphics[scale=0.3]{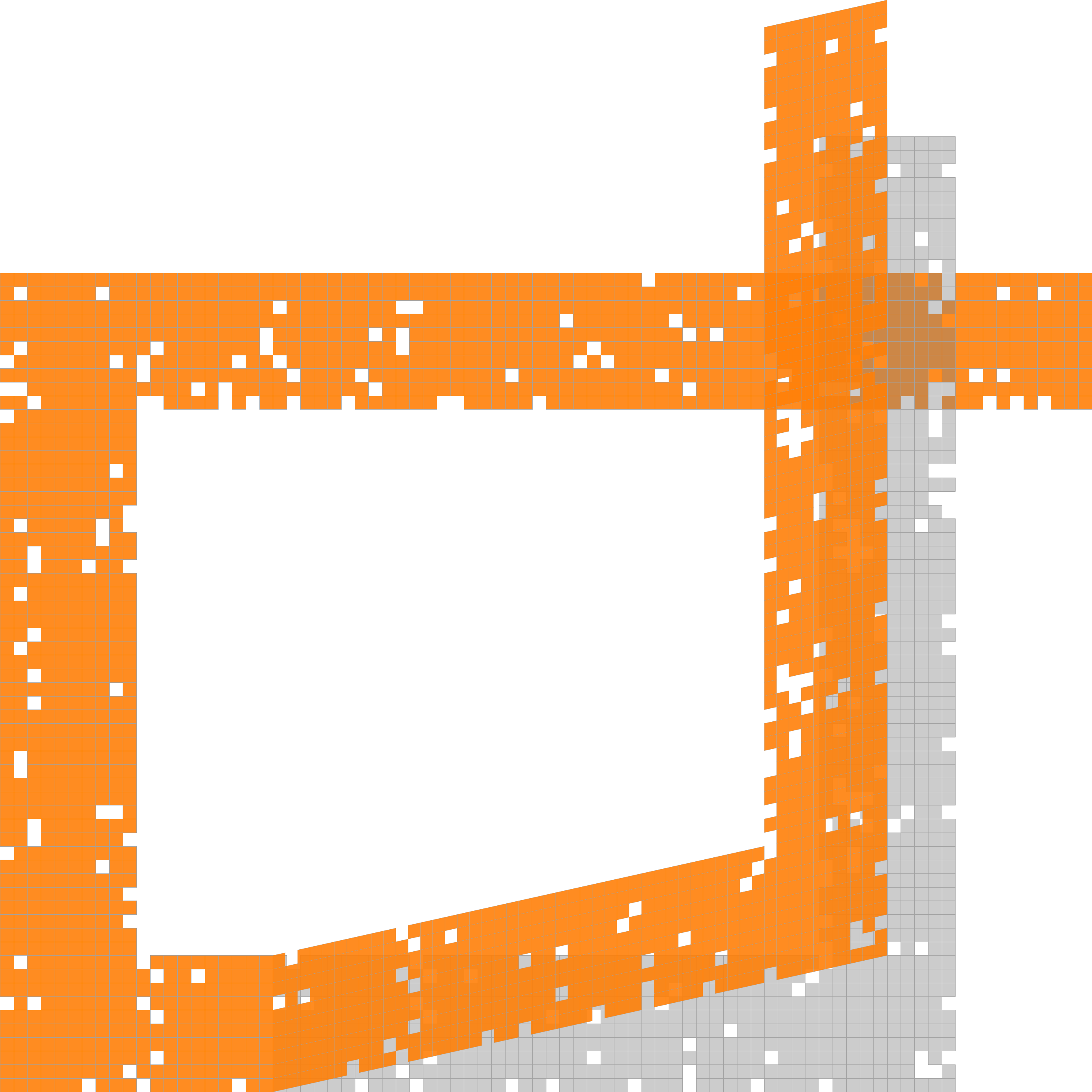}
      \caption{Assembly $\delta$ (in the cutspace)}
      \label{fig:delta}
    \end{figure}

    However, the projection of $\delta$ by $\pi$ hits some position twice, since the bottom and top arms would cross in $\mathbb{Z}^2$. Let thus $s^0$ be the longest prefix of $s$ yielding an assembly $\delta^0$ whose projection by $\pi$ hits each position in $\mathbb{Z}^2$ at most once.
    We claim that at least one tile placed by $s^0$ either has a mismatch with one of its neighbors, or else is placed with an excess of glue. Indeed, the tile it places at step $|s^0|$ must be overlapping some previously placed tile $s_i$. Let $j$ be the last index strictly before $|s^0|$ such that $s_j$ is a tile placement adjacent to $s_{|s^0|}$. This index exists because there must be at least one tile adjacent to $s_{|s^0|}$ before placing it.

    Moreover, $s_j$ must be in the same layer of the cut-space as $s_{|s^0|}$, because $s_{|s^0|}$ is not adjacent to an edge between the two layers. Therefore, $s_j$ is not in the same layer as $s_i$, and so these two tiles can be placed independently.
    Let now $s^1$ be the smallest prefix $s_{0,1,\ldots,\max(i,j)}$ of $s$ containing both tiles. If $\pi(s^1_{\max(i,j)})$ matches all its neighbors, then it can be placed with an excess of glue: indeed, the sum of the glue strengths with its neighbors in the cut-space is at least $\tau$, and it has a non-zero glue strength on its side in common with $\pi(s^1_{\min(i,j)})$ (since $s^1_{\max(i,j)}$ interacts with $s_{|s^0|}$).

\end{proof}

\section{3d Systems without mismatches simulate themselves}

We saw that systems without mismatches do not have the same power as general aTAM systems. We now investigate whether they have an intrinsically universal system, that is, whether there is a tileset on a given grid which can simulate all systems without mismatches on that same grid without itself making mismatches. In this section, we will show that in $\Z^3$, there is such a universal tileset for system without mismatches.%

\begin{reptheorem}{thm:full-hotel}
 There is a three dimensional tileset $U$ such that for every 3d tile assembly system $T$ without mismatches on $\Z^3$, there is a seed $s_U(T)$ such that tile assembly system $(U,s_U(T),2)$ (defined on $\Z^3$) simulates $T$ without making mismatches.
\end{reptheorem}

\subsection{Simulating 2d with 3d}

The full construction for Theorem \ref{thm:full-hotel} is somehow hard to visualize, as three-dimensional constructions tend to be. We will first prove a restricted version, in which the \emph{simulated} system is two-dimensional. The gist of both proofs is essentially the same.

\begin{theorem}\label{thm:2d-hotel}
There is a 3D tileset $U$ without mismatches that is intrinsically universal at temperature 2 for the class of 2D mismatch-free tile assembly systems.
\end{theorem}

\begin{proof}
In order to prove this, we rely on standard tools such as lookup tables and Turing machine simulations, as described in other constructions of intrinsically universal tilesets~\cite{IUSA,USA}.

One particular difficulty is to not create mismatches, which are used in crucial parts of the cited constructions. Although this is relatively straightforward when simulating \emph{deterministic} systems without mismatches, the main issue is to deal with non-determinism: indeed, if two different tiles can be placed at the same position by different cooperation, as shown on Figure~\ref{fig:diff}, the cut-space technique and the bisimilarity lemma could be used, like in the proof of Theorem~\ref{thm:nomismatches}, to reconcile inconsistent encodings of the north glue, and produce a ``cancerous'' glue.

    \newcommand\de{0.15}
    \newcommand\dw{0.15}
    \newcommand\hh{0.3}
\begin{figure}[ht]
  \centering
    \hfill
  \begin{tikzpicture}[scale=0.4]
    \draw(0,0) rectangle (4,4);
    \draw[draw=none, fill=black](2-\de,0)rectangle(2-\de+\dw,\hh);
    \draw[draw=none, fill=black](2+\de,0)rectangle(2+\de+\dw,\hh);
    \draw(2,\hh)node[anchor=south]{$a$};

    \draw[draw=none, fill=black](0,2-\dw/2)rectangle(\hh,2+\dw/2);
    \draw(\hh,2)node[anchor=west]{$b$};

    \draw[draw=none, fill=black](4-\hh,2-\dw/2)rectangle(4,2+\dw/2);
    \draw(4-\hh,2)node[anchor=east]{$d$};

    \draw[draw=none, fill=black](2-\de,4)rectangle(2-\de+\dw,4-\hh);
    \draw[draw=none, fill=black](2+\de,4)rectangle(2+\de+\dw,4-\hh);

    \draw(2,4-\hh)node[anchor=north]{$x$};
  \end{tikzpicture}
  \hfill
  \begin{tikzpicture}[scale=0.4]
    \draw(0,0) rectangle (4,4);
    \draw[draw=none, fill=black](2-\de,0)rectangle(2-\de+\dw,\hh);
    \draw[draw=none, fill=black](2+\de,0)rectangle(2+\de+\dw,\hh);
    \draw(2,\hh)node[anchor=south]{$a$};

    \draw[draw=none, fill=black](0,2-\dw/2)rectangle(\hh,2+\dw/2);
    \draw(\hh,2)node[anchor=west]{$b$};

    \draw[draw=none, fill=black](4-\hh,2-\dw/2)rectangle(4,2+\dw/2);
    \draw(4-\hh,2)node[anchor=east]{$d$};

    \draw[draw=none, fill=black](2-\de,4)rectangle(2-\de+\dw,4-\hh);
    \draw[draw=none, fill=black](2+\de,4)rectangle(2+\de+\dw,4-\hh);

    \draw(2,4-\hh)node[anchor=north]{$y$};
  \end{tikzpicture}
  \hfill
  \begin{tikzpicture}[scale=0.4]
    \draw[draw=none,fill=colb!15!white](0,0)rectangle(4,4);
    \draw(-1,4)--(0,4)--(0,-1);
    \draw(5,4)--(4,4)--(4,-1);
    \draw(-1,0)--(5,0);

    \draw[draw=none, fill=black](2-\de,0)rectangle(2-\de+\dw,-\hh);
    \draw[draw=none, fill=black](2+\de,0)rectangle(2+\de+\dw,-\hh);
    \draw(2,-\hh)node[anchor=north]{$a$};
    \draw[draw=none, fill=black](0,2-\dw/2)rectangle(-\hh,2+\dw/2);
    \draw(-\hh,2)node[anchor=east]{$b$};

    \draw[draw=none, fill=black](4+\hh,2-\dw/2)rectangle(4,2+\dw/2);
    \draw(4+\hh,2)node[anchor=west]{$d$};
  \end{tikzpicture}
  \hfill\vspace*{0pt}

  \caption{Both tiles on the right-hand side can attach to the empty location in blue on the assembly on the left-hand side, in different ways: for instance, $b$ and $d$ could cooperate and place the first tile, outputting an $x$ on the north side, or $a$ could place the second tile, outputting a $y$. The problem, if we stay in $\mathbb{Z}^2$, is that techniques such as those used in the proof of Theorem~\ref{thm:nomismatches} seem to allow us to ``mix' these two sequences in one, outputting an inconsistent encoding of a mix of $x$ and $y$ on the north side.}
  \label{fig:diff}
\end{figure}
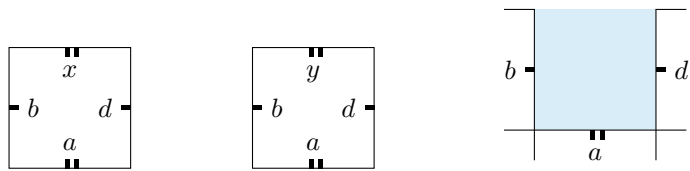

The main idea of this construction is to allow different assembly sequences to take place at the same time, but at different positions. The key to doing this is to embed the construction into $\mathbb{Z}^3$. In our construction, we use four layers: one to compute output sides, using different copies of the computational gadget for all possible input side combinations. One to choose (non-deterministically) the particular combination that will be used, preventing others to interfere with that choice, and two other layers to propagate information between all sides.

\subsubsection{Construction of a macrotile, step 1: gathering information}

The global layout of a macrotile is pictured on Figure~\ref{fig:hotel-layout-2d}: the lowest layer is the leftmost one, the highest is the rightmost one.
The construction of a macrotile \emph{starts at one of the intermediate layers}, when a neighboring tile places the encoding of $(k,g,T)$ on one of these two layers, where $k$ is the size of the gadgets (from which the scale factor is calculated), $g$ is the glue on the neighboring tile, and $T$ is the simulated tileset.

We take as an example the assembly of the first tile of Figure~\ref{fig:diff}, from the east: glue $d$ comes first, as an encoding of $(k,d,T)$ on the second layer (the H-wireboard layer). This encoding sends a ``probe'' containing only an encoding of $k$ and $T$, to the bottom layer. Then, a one-way wire is built, that goes around the bottom layer. This wire actually uses a counter to count up to the scale factor and turn at the appropriate places, until reaching the south side. The choice of the south as the preferred side is actually independent of where this probe started: all decisions are taken on the south side.

Since this counter does not depend on the types of glues that might come later around this supertile, all these glues will send the same probe; we can therefore design the probes, using appropriate counters, so that different probes grown concurrently have no mismatches when this happens.

When it reaches the south side, this probe  non-deterministically chooses $\log |T|$ bits. It then sends two different pieces of information to the two intermediate layers: (1) an ``unblocking'' wire consisting of only one tile, that instructs all neighboring glues to send information to the highest layer (where lookups are performed), and (2) the $\log |T|$ bits chosen.

Using these two pieces of information, all the gadgets in the highest layer that deal with east glues start to compute outputs: thanks to the $\log |T|$ chosen bits, this choice can be done uniformly among all gadgets.

Then, \emph{if another input arrives}, say, from the west, it will start the same probe to the bottom layer, where the probe and the already grown path agree, by our construction. In parallel to that, since the unblocking signal has already been sent by the bottom layer, the wires on the intermediate layers can start to grow, sending information to the gadgets on the highest layer, that can process the west side. These gadgets can finally start to compute potential outputs, using the encoding of $T$ and the $\log |T|$ ``non-deterministic'' bits. As detailed in the next section, all these gadgets compute different subsets of the input sides: therefore, only gadgets expecting only the east side will be able to start their computation, and, if $T$ can actually place a tile at the current position, will output the corresponding glues.

Clearly, the two intermediate layers can transmit all this information, both by having ``holes'' (to let the $\log |T|$ ``non-deterministic bits'' go through the layers), and by forming ``wires'' to connect glues to the highest layer, according to the routing scheme pictured by the full lines of Figure~\ref{fig:routeur}, and detailed later in the next section.

\subsubsection{Construction of a macrotile, step 2: computing the output}

The main purposes of the two intermediate layers, called \emph{H-wireboard} and \emph{V-wireboard} on Figure~\ref{fig:hotel-layout-2d} is to duplicate the information of the neighboring glues and of the $\log |T|$ non-deterministic bits chosen on the bottom layer, and to bring them to the gadgets on the top layer.

We describe these gadgets in greater detail now: they are actually made of standard lookup tables, as used for instance in~\cite{IUSA}; these constructions are described in the literature~\cite{Winfree98simulationsof,SolWin05} and do not use mismatches. We only need to arrange them properly in the space, so that neighboring gadgets do not interfere with each other. More precisely, we use the following routing scheme: if the supertile in construction has a supertile on its east side, that supertile sends two copies of the required information ($T$, the value of the glue, and the scale factor), drawn in purple on Figure~\ref{fig:routeur}, to the west. Similarly, when the supertile in construction has a tile on its west side, that supertile tile sends two copies of the same information to the east, but at different positions, shown in green on Figure~\ref{fig:routeur}.
A similar scheme is used to route glue information from the north and south sides (shown in blue and orange on Figure~\ref{fig:routeur}).

This information is then routed on their respective rows and colums, until it reaches gadgets that can handle it. On Figure~\ref{fig:routeur}, we can see that all combinations of input sides are indeed handled by different gadgets.

\begin{figure}[ht]
\centering
\begin{tikzpicture}[scale=0.2]
  \draw(0,0)rectangle(24,24);
  \begin{scope}[xshift=4cm,yshift=4cm]
    \draw(0,0)rectangle(16,16);
    \foreach \i in{4,8,12}{\draw[gray,very thin](0,\i)--(16,\i);}
    \foreach \i in{4,8,12}{\draw[gray,very thin](\i,0)--(\i,16);}
    \draw[cola,->](3,-4)--(3,16);
    \draw[cola,->](7,-4)--(7,16);
    \draw[colb,->](1,20)--(1,0);
    \draw[colb,->](9,20)--(9,0);

    \draw[colc,->](-4,13)--(16,13);
    \draw[colc,->](-4,5)--(16,5);

    \draw[cole,->](20,15)--(0,15);
    \draw[cole,->](20,11)--(0,11);

    \draw[cola,dashed](6.5,0)--(6.5,18);
    \draw[cola,dashed,->](14.5,0)--(14.5,18)--(3,18)--(3,20);
    \draw[cola,->,dashed](7,18)--(7,20);
    \draw[colb,dashed](9.5,16)--(9.5,-2);
    \draw[colb,->,dashed](13.5,16)--(13.5,-2)--(1,-2)--(1,-4);
    \draw[colb,->,dashed](9,-2)--(9,-4);
    \draw[cole,dashed](16,10.5)--(-2,10.5);
    \draw[cole,->,dashed](16,2.5)--(-2,2.5)--(-2,15)--(-4,15);
    \draw[cole,->,dashed](-2,11)--(-4,11);
    \draw[colc,->,dashed](0,1.5)--(18,1.5)--(18,13)--(20,13);
    \draw[colc,dashed](0,5.5)--(18,5.5);
    \draw[colc,->,dashed](18,5)--(20,5);
  \end{scope}
\end{tikzpicture}
\caption{Routing scheme: the outer rectangle is the macrotile, the inner one is the gadget area, divided into 16 gadgets. The inputs are routed by the full lines, and the outputs by the dashed line. On this figure, the two layers of routing (horizontal and vertical) are shown together. In the real construction, vertical and horizontal wires are exclusively in the V-wireboard and in the H-wireboard, respectively.}
\label{fig:routeur}
\end{figure}
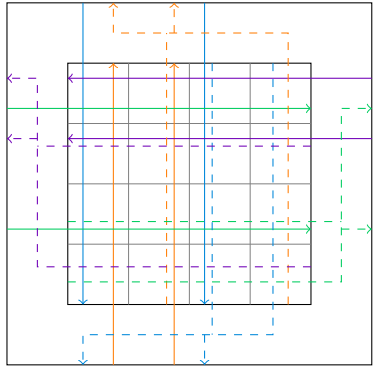

\subsubsection{Construction of a macrotile, step 3: outputting glues}

After computation takes place in the gadgets, it is time to construct glues on the remaining sides. To do so, we use the dashed lines shown on Figure~\ref{fig:routeur}. The information is completely output in each of the gadgets, routed by these lines to a ``peripheral'' area around the gadgets area, duplicated into two copies, and sent to the correct positions to start the next supertile.

The only issue that needs to be handled here, is the case where two different gadgets are activated, and send their result to the same wire (one of the wires in dashed line on Figure~\ref{fig:routeur}).
Since the $\log |T|$ bits chosen at stage 1 fixed the type of the supertile that will be placed (among the set of tiles that can attach there), the values sent on these wires are always equal: we can therefore simply use ``wires'' that can grow in both directions on these output rows, and remain consistent with other values that might be output by other gadgets.

Finally, because we are simulating a mismatch-free system, the glues output on abutting sides of two neighboring tiles are always equal, so we can design their encoding so that there is no mismatch at these positions.

\begin{figure}
  \centering
  \includegraphics[width=\textwidth]{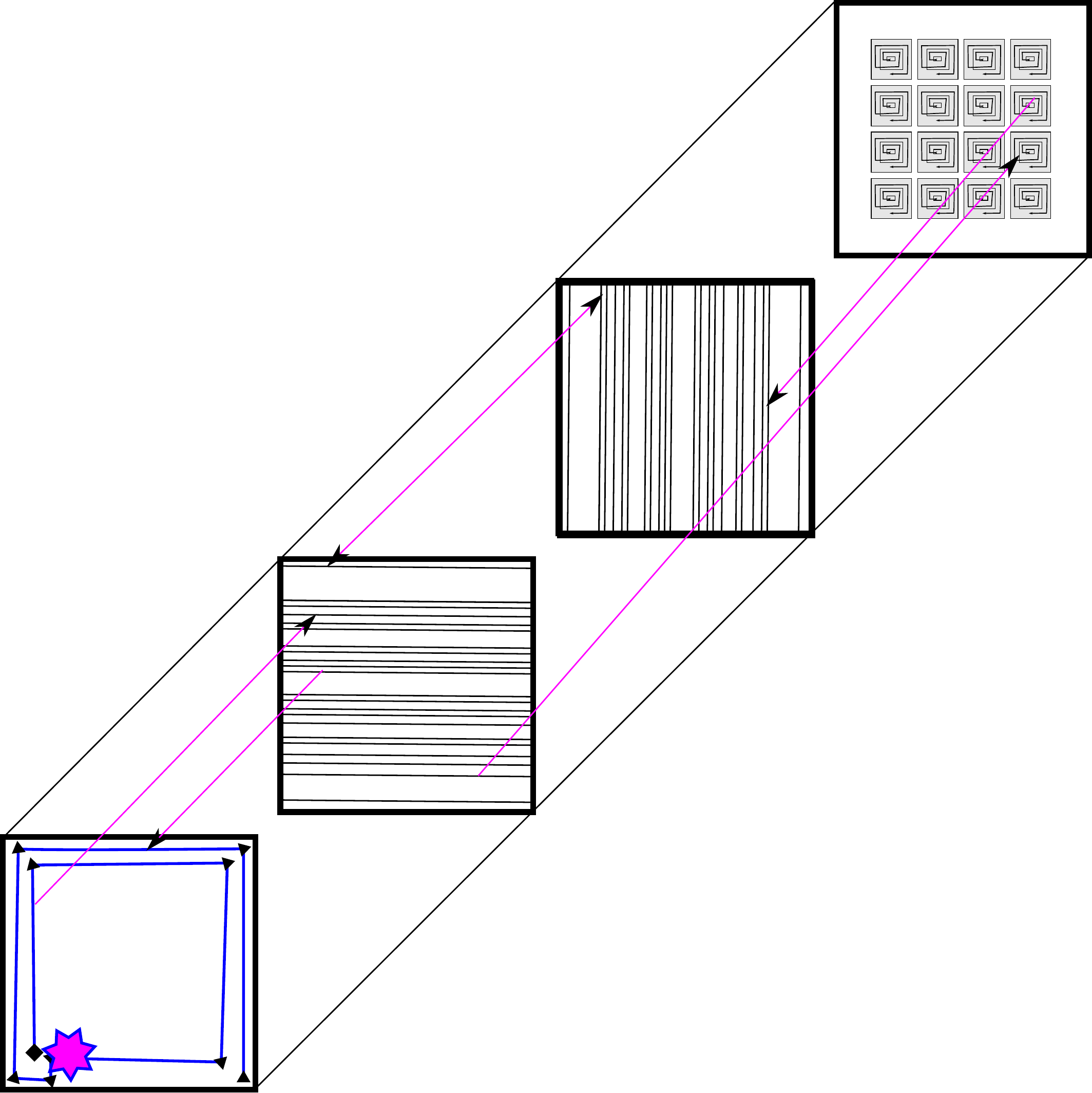}
  \caption{Layout of a macro-tile of $U$ simulating $T$. The top layer has 16 \emph{loci}, one for each possible combination of input sides, the middle layers are the wireboards, with the \emph{H-wireboard} carrying the E and W inputs and outputs between the edges and the loci, and the bottom layer, the \emph{V-wireboard} carrying the N and S inputs / outputs. Finally, the bottom layer is responsible for making a non-deterministic choice for the macrotile and distributing it to the system wires on the wireboards. This choice takes place in the starred location, after all input sides have had a chance to trigger the choice layer. The magenta arrows represent the inter-layer communication.}
  \label{fig:hotel-layout-2d}
\end{figure}

\end{proof}

\subsection{Simulating 3d with 3d}

We are now ready to give the proof of Theorem \ref{thm:full-hotel}.
\begin{reptheorem}{thm:full-hotel}
 There is a three dimensional tileset $U$ such that for every 3d tile assembly system $T$ without mismatches on $\Z^3$, there is a seed $s_U(T)$ such that tile assembly system $(U,s_U(T),2)$ (defined on $\Z^3$) simulates $T$ without making mismatches.
\end{reptheorem}

\begin{proof}
 Essentialy, the construction is the same, it just needs four times more instances of the gadget on the same plane, and, as a consequence, some changes in the wiring to accept input from the nadir and zenith. A scheme of this construction is shown on Figure~\ref{fig:hotel:3d}.

 To route the input from the nadir and zenith to the wireboards, we add one more instance of the routing scheme shown on Figure~\ref{fig:routeur}. These new wires need some space to reach the bottom of the gadgets: for that, we use six wires (three inputs, three outputs) on each side, instead of four in the previous construction.

\begin{figure}
  \centering
  \includegraphics[width=.4\textwidth]{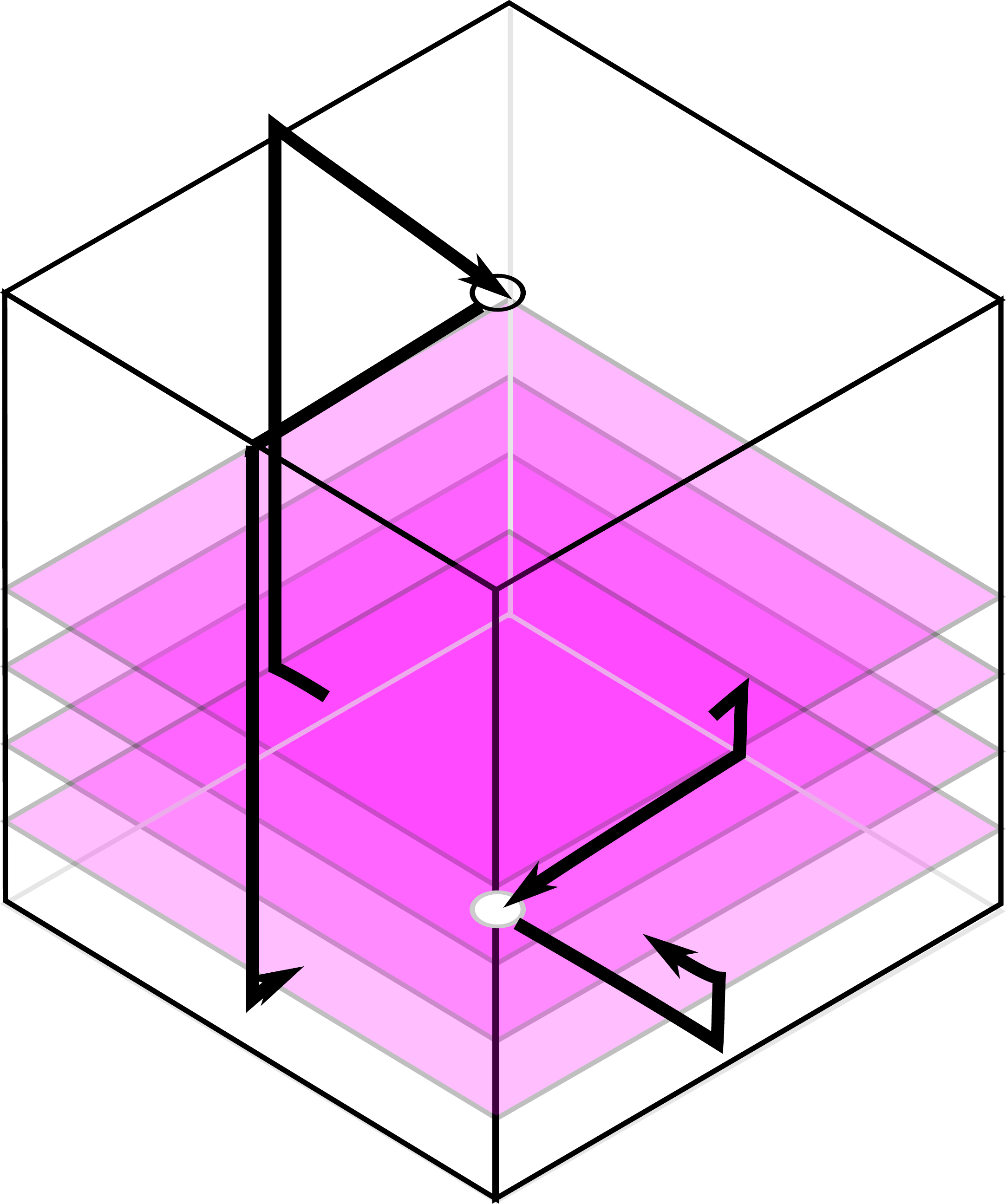}
  \caption{Layout of a macrotile; the shaded planes are the four planes corresponding to the previous simulation the wiring from the edges of these planes to the border of the macrotile is not represented, for clarity's sake.}
  \label{fig:hotel:3d}
\end{figure}

\end{proof}

\section{Acknowledgements}

The authors wish to thank Damien Woods for friendly and invaluable discussions, ideas, comments, improvements and support before and while writing this paper.

\end{document}